\documentclass[11pt,letterpaper]{article}
\usepackage[T1]{fontenc}
\usepackage[utf8]{inputenc}
\usepackage{array,tabularx}
\usepackage{tikz}
\usetikzlibrary{matrix,positioning,decorations.pathreplacing,calc}
\usepackage{algorithm}
\usepackage{algpseudocode}

\usepackage[margin=1in]{geometry}
\usepackage{amsmath,amssymb,amsthm,mathtools}
\usepackage{microtype}
\usepackage{lmodern}
\usepackage{enumitem}
\usepackage{booktabs}
\usepackage{xurl}
\usepackage[colorlinks=true,linkcolor=blue,citecolor=blue,urlcolor=blue]{hyperref}

\usepackage[backend=biber,style=alphabetic]{biblatex}
\newtheorem{theorem}{Theorem}[section]
\newtheorem{proposition}[theorem]{Proposition}
\newtheorem{lemma}[theorem]{Lemma}
\newtheorem{corollary}[theorem]{Corollary}

\theoremstyle{remark}
\newtheorem{fact}[theorem]{Fact}

\theoremstyle{plain}

\newcommand{\E}{\mathbf{E}}
\newcommand{\Pp}{\mathbf{P}}
\newcommand{\SYT}{\operatorname{SYT}}
\newcommand{\LIS}{\operatorname{LIS}}
\newcommand{\Pl}{\operatorname{Pl}}
\newcommand{\wt}{\operatorname{wt}}
\title{Exact Sampling of Permutations with a Prescribed\\ Longest Increasing Subsequence Length}
\author{Peter Clifford\\
	{\normalsize Department of Statistics}\\
	{\normalsize University of Oxford}\\
	{\normalsize United Kingdom}\\
	\and
	Rapha\"el Clifford\\
	{\normalsize Department of Computer Science}\\
	{\normalsize University of Bristol}\\
	{\normalsize  United Kingdom}\\
}
\date{}

\begin{document}
	\hypersetup{pageanchor=false}

	\maketitle
	\thispagestyle{empty}

	\begin{abstract}
		We study exact uniform sampling of permutations of length $n$ whose longest
		increasing subsequence (LIS) has prescribed length $k$.  Our main result is,
		to the best of our knowledge, the first polynomial-time exact sampler for this
		problem, valid for every $1\le k\le n$.  Via the Robinson--Schensted
		correspondence the problem reduces to sampling a Young diagram with first row
		of length $k$ from the hook-length-squared (conditioned Plancherel) law.  We
		sample this shape one coordinate at a time, where each conditional weight,
		a sum over exponentially many completions, collapses, through the
		Cauchy--Binet formula, to a single coefficient of the determinant of a small
		polynomial matrix.   A direct implementation runs in expected
		$\tilde O(n^4k^5)$ time in the word-RAM model. Aggregating the scores and exploiting the Hankel structure
		of the evaluated matrices reduces this to $\tilde O(n^2k^4)$. In the linear regime $k\in\Theta(n)$ we give a direct rejection sampler running in
		expected $O(n\log\log n)$ time, matching, up to constants, the cost of
		computing the LIS of one permutation.  For the relaxed
		constraint $\LIS(\pi)\le k$,
		plain rejection sampling gives expected $O(n\log\log n)$ time for every $k\ge4\sqrt n$, and so
		the worst case over all $k$ improves to $\tilde O(n^4)$.
	\end{abstract}

	\clearpage
	\hypersetup{pageanchor=true}
	\setcounter{page}{1}

	\section{Introduction}

	The longest increasing subsequence is one of the most intensively studied
	of all permutation statistics.  Interest in its behaviour on a uniformly
	random permutation goes back to Ulam, and its analysis, through
	Hammersley's interacting particle process, the limit-shape theorems of
	Logan--Shepp and Vershik--Kerov, and the Baik--Deift--Johansson fluctuation
	theorem, helped to build the modern theory connecting random
	permutations to random matrices and determinantal processes
	\cite{Hammersley1972,LoganShepp1977,VershikKerov1977,Johansson1998,BDJ1999,
	BorodinOkounkovOlshanski2000,Okounkov2000}.  See the papers of Aldous and
	Diaconis~\cite{AldousDiaconis1995,AldousDiaconis1999}, Stanley's
	survey~\cite{Stanley2006}, and Romik's book~\cite{romik2015surprising}
	for a comprehensive introduction.

	This paper asks a basic algorithmic question about this much-studied
	statistic. Can one
	efficiently sample a permutation of length $n$ \emph{uniformly at random}, subject to
	the constraint that its LIS has a prescribed length $k$?  Formally, for a
	permutation $\pi\in S_n$, let $\LIS(\pi)$ denote the length of its longest
	increasing subsequence.  Given input parameters $(n, k)$, we want to
	sample exactly from the uniform distribution on
	\[
	\Omega_{n,k}:=\{\pi\in S_n:\LIS(\pi)=k\}.
	\]

	\paragraph{Simple approaches fail.}
	A uniform permutation has $\LIS$ equal to a specific $k$ with probability
	that can be as small as $1/n!$, so rejection sampling from $S_n$ is hopeless
	in general.  The one exception is a narrow window around the typical value
	$k\approx2\sqrt n$, where plain rejection sampling runs in polynomial time. Elsewhere $\Pp(\LIS=k)$ is superpolynomially small.  A potentially viable route is to attempt
	Markov chain Monte Carlo on $\Omega_{n,k}$, but no chain with a provable
	polynomial mixing time is known and our goal is exact, not approximate,
	sampling.  The Robinson--Schensted correspondence reduces the problem to
	sampling a Young diagram with first row of length $k$ from a
	hook-length-squared distribution, but the number of feasible shapes, the
	partitions of $n-k$ into parts of size at most $k$, is polynomial only
	when $k$ or $n-k$ is bounded. This reaches $\exp(\Theta(\sqrt n))$ at
	$k=\Theta(\sqrt n)$.  We give a sampler whose running time is polynomial in both $n$ and $k$.

	\paragraph{Summary of results.}
	Our main result is the first
	polynomial-time exact sampler for $\Omega_{n,k}$, valid for every
	$1\le k\le n$. A direct implementation runs in expected $\tilde O(n^4k^5)$
	time (Theorem~\ref{thm:intro-general-direct}), and a structured
	implementation of the \emph{same} sampler runs in expected $\tilde O(n^2k^4)$
	time (Theorem~\ref{thm:intro-general-fast}).  Since $k\le n$, this gives an $\tilde O(n^6)$-time exact sampler for every $k$.   In the linear
	regime $k\in\Theta(n)$ a much faster solution is possible. Rejection sampling on an expanded proposal space achieves
	expected $O(n\log\log n)$ time (Theorem~\ref{thm:intro-largek}). This
	matches, up to constants, the best known bound for merely \emph{computing}
	the LIS of a single permutation~\cite{CrochemorePorat2010}, so in this regime
	sampling a uniform element of $\Omega_{n,k}$ is as cheap as testing
	membership in it.  The same sampler degrades gracefully, remaining
	polynomial-time for all $k\ge\eta n/\log n$ (Corollary~\ref{cor:sublinear-k}). Table~\ref{tab:results} summarises these bounds, and
	Section~\ref{sec:results} states the results formally.  The general sampler extends, with the same bounds, to the relaxed constraint
	$\LIS(\pi)\le k$ (Corollary~\ref{cor:atmost}). For this relaxed problem,
	plain rejection sampling from $S_n$ already gives expected $O(n \log \log n)$ time for every $k\ge4\sqrt n$, and so
	the worst case over all $k$ improves to $\tilde O(n^4)$.

	\paragraph{Techniques.}
	The main sampler rests on the Robinson--Schensted correspondence, which
	reduces uniform sampling from $\Omega_{n,k}$ to sampling a Young diagram
	with first row of length $k$ from Plancherel measure conditioned on
	$\lambda_1=k$.  We conjugate the shape and apply a coordinate shift,
	re-expressing this law as a fixed-sum distribution on strict partitions
	whose weight is a squared Vandermonde factor times independent factorial
	weights.  The sampler then exposes the shifted coordinates one at a time.
	At each partial state, it chooses the next coordinate with probability
	proportional to the total weight of all valid completions.
	
	The completion sums range over exponentially many suffixes, but they are
	not enumerated.  By Cauchy--Binet, each such sum is a coefficient of the
	determinant of a polynomial moment matrix.  At each scalar evaluation this
	moment matrix is Hankel, so the needed determinant values can be computed
	from $O(k)$ moments in near-linear time.  A second saving comes from
	the observation that the total score of all candidates up to a cutoff is itself a single determinant coefficient, so each coordinate can be chosen with $O(\log n)$ determinant evaluations instead of one per candidate.  The
	resulting completion scores
	are exact integers. They are computed modulo word-size primes and
	reconstructed by Chinese remaindering, so no rounding enters the
	procedure and the output law is exactly uniform.
	
	Conceptually, this is an extension-counting sampler in the sense of
	Jerrum, Valiant and Vazirani~\cite{JerrumValiantVazirani1986}. Exact
	counts of completions from each partial state determine the conditional
	distribution of the next step.  The new ingredient here is the
	determinant-based oracle that makes these completion counts efficient.

	For the linear regime where $k \in \Theta(n)$ a lighter idea suffices.  We sample from an expanded
	space of pairs $(\pi,I)$ in which $I$ is a distinguished $k$-term increasing
	subsequence of the permutation $\pi$, and accept only when $I$ is the \emph{leftmost} LIS of
	$\pi$. This is a canonical witness computable in $O(n\log\log n)$ time.  Every
	permutation with $\LIS=k$ then survives in exactly one accepted pair, so
	acceptance yields exact uniformity, and a Plancherel-measure tail bound shows
	that the acceptance probability is bounded below by a constant.  This same rejection sampling idea gives a polynomial-time sampling algorithm as long as $k \in \Omega(n/\log n)$.

	\paragraph{Related work.}
	On static inputs, the computation of the LIS admits an
	$O(n\log n)$ comparison-based algorithm~\cite{Fredman1975}. For
	permutations of $[n]:=\{1,\dots,n\}$, faster word-RAM algorithms are known. In
	particular, the standard predecessor-based implementation gives
	$O(n\log\log n)$ time, and Crochemore and Porat give an
	$O(n\log\log \ell)$ bound in terms of the LIS length $\ell$
	\cite{CrochemorePorat2010}. Later work developed streaming and
	sublinear-time approximation algorithms for $\LIS$ and distance to
	monotonicity
	\cite{GopalanJayramKrauthgamerKumar2007,ErgunJowhari2008,SaksSeshadhri2017,AndoniNosatzkiSinhaStein2022}.
	On dynamic inputs, variants of the problem have been studied in
	sliding-window and fully dynamic models
	\cite{AlbertGolynskiHamelLopezOrtizRaoSafari2004,ChenChuPinsker2013,MitzenmacherSeddighin2020,KociumakaSeddighin2021}.
	These results concern computing or approximating the LIS of a given
	input. Our problem is orthogonal: we seek to generate a random
	permutation conditioned on having a prescribed LIS length.
	
	The closest exact-sampling results we are aware of are also based on Robinson--Schensted-type correspondences and Young tableaux, but they sample different ensembles. Thomas and Yong use Hecke insertion to give an exact polynomial-time sampler for a Plancherel-type measure, while Betea, Boutillier, Bouttier, Chapuy, Corteel, and Vuleti\'c give perfect samplers for Schur processes \cite{ThomasYong2011,BeteaBoutillierBouttierChapuyCorteelVuletic2018}. These samplers produce Plancherel- or Schur-type random shapes in which the total size, and in particular the first row length, are not fixed in advance. Uniform sampling from $\Omega_{n,k}$, however, requires the two hard constraints $|\lambda|=n$ and $\lambda_1=k$. One could try to recover these constraints by rejection, but the probability of simultaneously obtaining the prescribed size and first row length is too small in the regimes considered here. Our sampler instead enforces the condition $\lambda_1=k$ by construction. To the best of our knowledge, exact uniform sampling from $\Omega_{n,k}$ has not previously been studied.
	
	By contrast, \emph{counting} $\Omega_{n,k}$ is polynomial time by standard means.  Gessel's identity expresses the generating function of the counts $|\Omega_{n,\le k}|$ as a $k\times k$ Toeplitz determinant of modified Bessel functions~\cite{Gessel1990}, and extracting the coefficient of $x^{2n}$ from this determinant, computed over a truncated power-series ring with exact arithmetic, gives $|\Omega_{n,k}|=|\Omega_{n,\le k}|-|\Omega_{n,\le k-1}|$ in polynomial time.  Counting, however, returns only a single aggregate.  Exact sequential sampling requires, at every partial state, the total weight of all valid completions of that state, a conditional refinement that the aggregate identity does not provide and that is the role of the determinant oracle of Section~\ref{sec:oracle}.

	\section{Our Results}\label{sec:results}

	Table~\ref{tab:results} summarises the running times established in this
	paper.

	\begin{table}[h]
		\centering
		\small
		\setlength{\tabcolsep}{3pt}
		\begin{tabularx}{\textwidth}{@{}>{\raggedright\arraybackslash}p{0.25\textwidth}>{\raggedright\arraybackslash}p{0.29\textwidth}X@{}}
			\toprule
			Regime & Expected time & Technique \\
			\midrule
			$1\le k\le n$ & $\tilde O(n^4k^5)$ & determinant oracle, direct evaluation \\
			$1\le k\le n$ & $\tilde O(n^2k^4)$ & cumulative scores, Hankel moments, superfast determinants \\
			$k\in\Theta(n)$ & $O(n\log\log n)$ & rejection on an expanded proposal space \\
			$k\ge \eta n/\log n$ & $O(n^{1+2/\eta}\log\log n)$ & expanded rejection (Corollary~\ref{cor:sublinear-k}) \\
			$\LIS(\pi)\le k$, all $1\le k\le n$ & $\tilde O(n^2k^4)$; $O(n\log\log n)$ for $k\ge4\sqrt n$ & Corollary~\ref{cor:atmost} \\
			\bottomrule
		\end{tabularx}
		\caption{Samplers are exact and bounds are expected times in the word-RAM model.}
		\label{tab:results}
	\end{table}

	Our main contribution is a polynomial-time exact sampler for every
	$1\le k\le n$.  By the Robinson--Schensted correspondence (recalled in
	Section~\ref{sec:general}) it reduces to sampling a partition $\lambda\vdash n$
	with $\lambda_1=k$ from $\Pp_{n,k}(\lambda)\propto(f^\lambda)^2$, and then
	sampling a uniform pair of tableaux of that shape.  The shape sampler is the
	algorithmic core of the paper. The running times of its two implementations are stated next.

	\begin{theorem}[General exact sampler, direct implementation]
	\label{thm:intro-general-direct}
	For every $1\le k\le n$, there is an exact sampler for the uniform
	distribution on $\Omega_{n,k}$. In the word-RAM model, the direct
	implementation has expected running time
	\[
	\tilde O\!\bigl(n^4k^5\bigr).
	\]
	\end{theorem}

	\begin{theorem}[General exact sampler, fast implementation]
		\label{thm:intro-general-fast}
		For every $1\le k\le n$, the same sampler can be implemented in
		\[
		\tilde O\!\bigl(n^2k^4\bigr)
		\]
		expected time, by drawing each coordinate through cumulative completion
		scores, evaluated via batched multipoint evaluation of Hankel moment
		sequences and a superfast Hankel determinant
		algorithm over finite fields.
	\end{theorem}

	The sampler underlying both theorems is developed in
	Section~\ref{sec:general}. Correctness is proved in
	Theorem~\ref{thm:shape-correct} and Lemma~\ref{lem:shape-to-perm} and the
	running-time analyses are given in Sections~\ref{sec:direct-time} and~\ref{sec:complexity}.  Since
	$k\le n$, Theorem~\ref{thm:intro-general-fast} gives an $\tilde O(n^6)$ exact
	sampler for every $k$.

	For the linear regime $k \in \Theta(n)$ we give a much simpler and faster sampler.

	\begin{theorem}[Large-$k$ direct sampler]\label{thm:intro-largek}
	Suppose $k\in\Theta(n)$.  There is an exact sampler for the
	uniform distribution on $\Omega_{n,k}$ whose expected running time in the
	word-RAM model is $O(n\log\log n)$, where the implicit constant may depend on
	the constants in $k\in\Theta(n)$.
	\end{theorem}

	This sampler is described in Section~\ref{sec:direct}, with its full analysis
	in Section~\ref{sec:linear}.  The same sampler stays polynomial well below the
	linear regime. For every fixed $\eta>0$, if $k\ge\eta n/\log n$ its expected
	running time is $O(n^{1+2/\eta}\log\log n)$
	(Corollary~\ref{cor:sublinear-k}).  Finally, all of the above extend to the
	relaxed constraint $\LIS(\pi)\le k$.

	\begin{corollary}\label{cor:atmost}
	For every $1\le k\le n$ there is an exact sampler for the uniform
	distribution on $\Omega_{n,\le k}:=\{\pi\in S_n:\LIS(\pi)\le k\}$ with
	expected running time $\tilde O(n^2k^4)$ in the word-RAM model
	($\tilde O(n^4k^5)$ for the direct implementation).  Moreover, for
	$k\ge 4\sqrt n$, rejection sampling directly from $S_n$ is an exact sampler
	for $\Omega_{n,\le k}$ with expected running time $O(n\log\log n)$.
	Consequently the uniform distribution on $\Omega_{n,\le k}$ can be sampled
	exactly in expected time $\tilde O(n^4)$ for every $1\le k\le n$.
	\end{corollary}

	Corollary~\ref{cor:atmost} is proved in Section~\ref{sec:atmost}.

	\paragraph{Model of computation.}
	We work in the randomised word-RAM model with word size $w=\Theta(\log n)$. A uniform
	random word, $\Theta(\log n)$ unbiased random bits, is produced in $O(1)$
	time. We write $\tilde O(\cdot)$ for $O(\cdot)$ up to factors polylogarithmic in
	$n$. Weighted choices with integer weights are made exactly. Given nonnegative integer weights $a_1,\dots,a_M$ of sum $W$, draw a
	uniform integer in $\{1,\dots,W\}$ by rejection sampling from the next power of two and
	locate it among the prefix sums by binary search, at expected cost $O(\log W)$
	random bits and $O(\log M)$ comparisons of $O(\log W)$-bit integers.

	\section{A polynomial-time algorithm for general \texorpdfstring{$k$}{k}}\label{sec:general}

	By the Robinson--Schensted correspondence, sampling a uniform permutation
	from $\Omega_{n,k}$ reduces to two steps.  The correspondence is a bijection
	between permutations of $[n]$ and pairs $(P,Q)$ of standard Young tableaux of
	a common shape $\lambda$, and it identifies $\LIS(\pi)$ with the first row
	length $\lambda_1$. Figure~\ref{fig:rs-example} shows an example.  The number
	of permutations of shape $\lambda$ equals $(f^\lambda)^2$, where
	$f^\lambda=|\SYT(\lambda)|$ \cite{Schensted1961}.  Thus uniform sampling from
	$\Omega_{n,k}$ is equivalent to first sampling a partition $\lambda\vdash n$
	from
	\[
	\Pp_{n,k}(\lambda)\propto (f^\lambda)^2\,\mathbf 1_{\{\lambda_1=k\}},
	\]
	where $\mathbf 1$ denotes the indicator function,
	and second, sampling two independent uniform tableaux of shape $\lambda$ and
	applying the inverse correspondence.  The second step is standard and is
	deferred to Section~\ref{sec:tableaux}. The first step is the main
	algorithmic challenge, since a naive dynamic program over Young diagrams may
	have superpolynomially many states.

	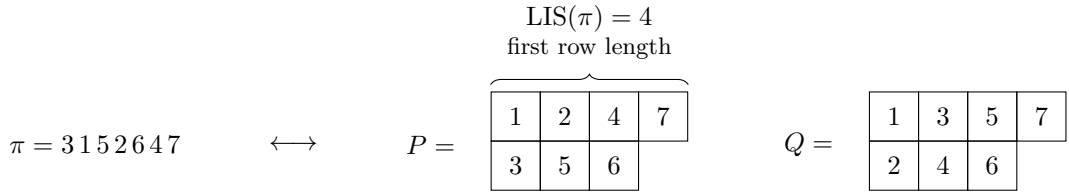
\begin{figure}[h]
		\centering
		\begin{tikzpicture}[
			>=latex,
			every node/.style={font=\small},
			box/.style={draw, minimum width=6.5mm, minimum height=6.5mm, inner sep=0pt},
			brace/.style={decorate, decoration={brace, amplitude=4pt}}
			]

			\node (perm) {$\pi = 3\,1\,5\,2\,6\,4\,7$};

			\node[right=0.9cm of perm] (arrow)
			{$\longleftrightarrow$};

			\node[right=0.9cm of arrow, anchor=west] (Plabel) {$P=$};
			\matrix (P) [matrix of nodes,
			nodes={box},
			row sep=-\pgflinewidth,
			column sep=-\pgflinewidth,
			anchor=west,
			right=2mm of Plabel.east] {
				1 & 2 & 4 & 7 \\
				3 & 5 & 6 & |[draw=none]| \\
			};

			\node[right=1.0cm of P.east, anchor=west] (Qlabel) {$Q=$};
			\matrix (Q) [matrix of nodes,
			nodes={box},
			row sep=-\pgflinewidth,
			column sep=-\pgflinewidth,
			anchor=west,
			right=2mm of Qlabel.east] {
				1 & 3 & 5 & 7 \\
				2 & 4 & 6 & |[draw=none]| \\
			};

			\draw[brace]
			($(P-1-1.north west)+(0,2.5pt)$) --
			($(P-1-4.north east)+(0,2.5pt)$)
			node[midway, above=6pt, align=center]
			{$\LIS(\pi)=4$\\[-1pt]\footnotesize first row length};

		\end{tikzpicture}
		\caption{An example of the Robinson--Schensted correspondence. The permutation
			$\pi=3152647$ corresponds to a pair $(P,Q)$ of standard Young tableaux of common
			shape. The first row length is \(4\), which equals \(\LIS(\pi)\).}
		\label{fig:rs-example}
	\end{figure}

	\paragraph{Overview of the shape sampler.}
The challenge is to sample the Robinson--Schensted shape conditioned on
$\lambda_1=k$ with weight $(f^\lambda)^2$, where direct enumeration over the up
to $\binom nk$ feasible shapes is not polynomial.  Three ideas remove the obstacles in turn. These are a reparameterisation that separates the weight, a determinantal identity that collapses the completion sum, and the Hankel structure that makes the resulting determinant computation cheap. Conjugating the partition and applying
the shift $x_i:=\rho_i+k-i$ rewrites the hook-length formula as
$(f^\lambda)^2\propto\Delta(x)^2\prod_{i=1}^k x_i!^{-2}$, separating the
per-coordinate factorial weights from the squared Vandermonde coupling
$\Delta(x)^2$. This separated weight admits a one-coordinate-at-a-time sampler
that, given a prefix $p$, draws the next coordinate $t$ with probability
proportional to the total weight of all completions. The obstacle is that each
such \emph{completion score} sums over up to exponentially many suffixes. By the
Cauchy--Binet formula that sum collapses to a single coefficient of the
determinant of a polynomial matrix, computable in polynomial time. When the
matrix is evaluated at a scalar~$\xi$ it is \emph{Hankel}. Its determinant
costs $\tilde O(L)$ field operations rather than $O(L^3)$, which drives the fast implementation.

The rest of this section formalises the weight, defines the completion scores,
gives the determinant oracle, and proves correctness of the sequential shape
sampler. Section~\ref{sec:tableaux} completes the reduction to permutations,
Section~\ref{sec:direct-time} bounds the direct implementation, and
Section~\ref{sec:complexity} develops the fast implementation.

	\subsection{Reparameterisation}\label{sec:reparam}

	We first recall the hook-length formula.  For a partition $\nu$ and a cell
	$u=(i,j)$ of its Young diagram, let $h_\nu(u)=\nu_i-j+\nu'_j-i+1$ denote the
	hook length of $u$ (the number of cells consisting of $u$, those to its right
	in row $i$, and those below it in column $j$), where $\nu'$ is the conjugate
	partition whose parts are the column lengths of $\nu$.  Write
	$H_\nu:=\prod_{u\in\nu}h_\nu(u)$.  The hook-length
	formula states that $f^\nu=|\nu|!/H_\nu$, where $f^\nu=|\SYT(\nu)|$
	\cite{FrameRobinsonThrall1954}.

	We work with the conjugate partition $\rho=\lambda'$ rather than $\lambda$
	directly. Its parts $\rho_i$ are the column heights of $\lambda$, and it is
	these that appear naturally in the hook-length formula. Since $\lambda_1=k$,
	the partition $\rho$ has exactly $k$ positive parts, where $\rho_i$ counts
	the rows of $\lambda$ of length at least $i$, satisfying
	\[
	\rho_1\ge\rho_2\ge\cdots\ge\rho_k\ge 1,\qquad \rho_1+\cdots+\rho_k=n.
	\]
	Set $x_i:=\rho_i+k-i$ and $N:=n+\binom{k}{2}$. The shift turns the weakly
	decreasing sequence $\rho$ into the strictly decreasing sequence
	\[
	n\ge x_1>x_2>\cdots>x_k\ge 1,\qquad x_1+\cdots+x_k=N,
	\]
	and the map $x\mapsto\rho$ given by $\rho_i=x_i-k+i$ is its inverse. Thus the
	shapes $\lambda\vdash n$ with $\lambda_1=k$ are in bijection with the strict
	partitions $x$ of $N$ into exactly $k$ parts, each at most $n$. This is the same shift Johansson used to express Plancherel measure as a squared Vandermonde factor times independent weights~\cite{Johansson2001}. Here we repurpose it for purely algorithmic purposes and apply it to the conjugate partition. Concretely, for the shape of Figure~\ref{fig:rs-example}, where $n=7$ and $k=4$, we have $\lambda=(4,3)$, $\rho=(2,2,2,1)$, $x=(5,4,3,1)$, and $N=7+\binom{4}{2}=13$.

	For a strict partition $x=(x_1,\dots,x_k)$, write
	$\Delta(x):=\prod_{1\le i<j\le k}(x_i-x_j)$. More generally, for any
	strictly decreasing sequence, $\Delta$ denotes the same product over its
	entries.

	The following lemma makes the factored form of the target weight explicit.
	It is the reason the $x$ coordinates are useful. The hook-length formula,
	when rewritten in terms of $x$, separates into a Vandermonde factor and
	independent factorial weights, yielding a target distribution that admits
	an efficient sequential sampler.

	\begin{lemma}\label{lem:x-hook}
		Under the bijection above,
		\[
		f^\rho=n!\,\frac{\Delta(x)}{\prod_{i=1}^k x_i!}.
		\]
		Consequently the target distribution on strict partitions \(x\) is proportional to
		\[
		\wt(x):=\Delta(x)^2\prod_{i=1}^k x_i!^{-2}.
		\]
	\end{lemma}

	\begin{proof}
	We show that each row's hook lengths fill the run $1,2,\ldots$ up to the
	longest hook, with known gaps, so its hook product is a factorial divided
	by the product of the gap values. Multiplying over the rows then gives the
	product formula for $H_\rho$.
	By the hook-length formula $f^\rho=n!/H_\rho$. For row $i$ of $\rho$, the hook
	lengths $h_\rho(i,c)=\rho_i-c+\rho'_c-i+1$ as $c$ ranges from $1$ to $\rho_i$ give
	the integers $1,\dots,\rho_i+k-i$ with the values $\rho_i-\rho_j+j-i$ for
	$j>i$ omitted. To see this, compare consecutive columns. For
	$1\le c<\rho_i$ we have $h_\rho(i,c)-h_\rho(i,c+1)=1+\delta_c$, where
	$\delta_c:=\rho'_c-\rho'_{c+1}$ counts the rows of length exactly $c$.
	These rows are $l=\rho'_{c+1}+1,\dots,\rho'_c$, and the $\delta_c$ values
	skipped at this step, namely $h_\rho(i,c+1)+1,\dots,h_\rho(i,c)-1$, equal
	$\rho_i-\rho_l+l-i$ for exactly those rows $l$. The smallest hook is
	$h_\rho(i,\rho_i)=\rho'_{\rho_i}-i+1$, and the values
	$1,\dots,\rho'_{\rho_i}-i$ beneath it equal $\rho_i-\rho_l+l-i$ for the
	rows $l>i$ of length $\rho_i$. Since $h_\rho(i,1)=\rho_i+k-i$, the hook
	lengths of row $i$ and the values $\rho_i-\rho_j+j-i$ for $j>i$ together
	fill $\{1,\dots,\rho_i+k-i\}$. Hence
	\[
	\prod_{c=1}^{\rho_i}h_\rho(i,c)
	=\frac{(\rho_i+k-i)!}{\prod_{j>i}(\rho_i-\rho_j+j-i)}.
	\]
	Multiplying over $i=1,\dots,k$ gives
	\[
	H_\rho
	=\frac{\prod_{i=1}^k(\rho_i+k-i)!}{\prod_{1\le i<j\le k}(\rho_i-\rho_j+j-i)}.
	\]
	Substituting $x_i=\rho_i+k-i$ yields $(\rho_i+k-i)!=x_i!$ and
	$\rho_i-\rho_j+j-i=x_i-x_j$, so $H_\rho=\prod_i x_i!/\Delta(x)$ and
	$f^\rho=n!\,\Delta(x)/\prod_i x_i!$.

	Since transposition preserves the number of standard Young tableaux,
	$f^\lambda=f^\rho$. Therefore
	\[
	(f^\lambda)^2=(f^\rho)^2=(n!)^2\Delta(x)^2\prod_i x_i!^{-2}.
	\]
	As $(n!)^2$ is independent of $x$, the induced distribution on $x$ is
	proportional to $\wt(x)$.
	\end{proof}

	Lemma~\ref{lem:x-hook} reduces uniform sampling on $\Omega_{n,k}$ to
	sampling a strict partition $x$ with probability proportional to $\wt(x)$,
	followed by the Robinson--Schensted construction of a permutation from a
	random pair of tableaux of the corresponding shape. Algorithm~\ref{alg:general-sampler}
	assembles these steps. The remainder of this section is devoted to its
	central ingredient, the strict-partition sampler of
	Algorithm~\ref{alg:shape-sampler}.

	\begin{algorithm}[h]
		\caption{General exact sampler}
		\label{alg:general-sampler}
		\begin{algorithmic}[1]
			\Require Integers $n,k$ with $1\le k\le n$
			\State Set $N\gets n+\binom{k}{2}$.
			\State Sample a strict partition $x$ of $N$ with $k$ parts via
			Algorithm~\ref{alg:shape-sampler}, with
			$\Pp(x)\propto\Delta(x)^2\prod_i x_i!^{-2}$.
			\State Recover $\rho_i\gets x_i-k+i$ for $i=1,\dots,k$.
			\State Set $\lambda\gets\rho'$.
			\State Sample independently two uniform tableaux $P,Q\in\SYT(\lambda)$.
			\State Return the inverse Robinson--Schensted image of $(P,Q)$.
		\end{algorithmic}
	\end{algorithm}

	\subsection{Completion scores}\label{sec:scores}

	We sample the strict partition $x_1>\cdots>x_k$ one coordinate at a time. At each step, having fixed a prefix
	$p=(p_1,\dots,p_r)$ with $p_1>\cdots>p_r\ge 1$, we choose the next coordinate
	$t$ with probability proportional to the total $\wt$-mass of all valid
	extensions of $(p_1,\dots,p_r,t)$ to a full strict partition. These masses are
	the completion scores $S_p(t)$, which we define and compute below.

	At a prefix $p=(p_1,\ldots,p_r)$ of length $r<k$, strictness means the next coordinate can be at most
	\[
	B(p):=
	\begin{cases}
		n, & r=0,\\
		p_r-1, & r>0.
	\end{cases}
	\]
	The admissible candidates for the next coordinate are then
	$t\in\{1,\ldots,B(p)\}$.

	For such an admissible $t$, write
	\[
	L:=k-r-1,\quad
	s(t):=N-\textstyle\sum_{i=1}^r p_i-t,\quad
	E_t:=\{1,\dots,t-1\}.
	\]
	Here $L$ is the number of coordinates that remain to be chosen after $t$,
	$s(t)$ is the sum they must contribute so that the whole sequence sums to
	$N$, and $E_t$ is the set of values available to them.
	A valid suffix is an $L$-element subset $A\subseteq E_t$ with $\sum A=s(t)$.
	For a finite set $A$, we use the same positive Vandermonde convention,
	\[
	\Delta(A):=\prod_{\substack{a,b\in A\\a<b}}(b-a),
	\]
	with $\Delta(\varnothing)=1$.

	Since the sampler uses only relative probabilities at a fixed prefix, we may
	scale all completion masses by any common positive factor. We clear the
	factorial denominators of the non-prefix coordinates by multiplying through
	by $(n!)^{2(L+1)}$. For a possible later value $a$, define
	\[
	G_p(a):=\Bigl(\frac{n!}{a!}\Bigr)^2\prod_{i=1}^r(p_i-a)^2,
	\]
	recording the scaled contribution of $a$ and its pairwise interactions with
	the prefix. Fixing the next coordinate to $t$ introduces a further pairwise
	factor $(t-a)^2$ for each suffix element $a$, giving
	\[
	\psi_{p,t}(a):=(t-a)^2\,G_p(a).
	\]
	The completion score of a candidate $t$ is then the total scaled mass of its
	valid suffixes,
	\[
	S_p(t):=G_p(t)
	\sum_{\substack{A\subseteq E_t\\|A|=L,\;\sum A=s(t)}}
	\Delta(A)^2\prod_{a\in A}\psi_{p,t}(a),
	\qquad 1\le t\le B(p).
	\]
	The sum over $A$ is an instance of a general quantity, which we isolate so
	that the determinant oracle of Section~\ref{sec:oracle} can evaluate it
	without enumerating the suffixes. For a ground set $E$ and a weight function
	$\psi$ on $E$, define the completion counter
	\[
	Z_{L,s}(E,\psi)
	:=\sum_{\substack{A\subseteq E\\|A|=L,\;\sum A=s}}
	\Delta(A)^2\prod_{a\in A}\psi(a),
	\]
	with $Z_{0,s}(E,\psi)=\mathbf{1}_{s=0}$ since the only suffix of length zero
	is $A=\varnothing$. In this notation,
	\[
	S_p(t)=G_p(t)\,Z_{L,s(t)}(E_t,\psi_{p,t}).
	\]

	The scores $S_p(t)$ are nonnegative integers, since each $G_p(a)$ is the
	square of an integer and each $Z_{L,s(t)}(E_t,\psi_{p,t})$ is a sum of
	products of such integers. Thus the weighted choice in
	Algorithm~\ref{alg:shape-sampler} can be implemented exactly with integer
	weights, with no rational probabilities or rounding.

	\begin{lemma}\label{lem:score-mass}
		For a fixed prefix $p$, the completion score $S_p(t)$ is proportional to the
		total $\wt$-mass of all valid extensions of $p$ with next coordinate $t$, with
		a proportionality constant independent of $t$.
	\end{lemma}

	\begin{proof}
		We use the fact that the squared Vandermonde of a full extension factors into
		prefix-only, prefix--candidate, suffix-only, candidate--suffix, and
		prefix--suffix terms.  After the prefix-only factor is pulled out and one
		factor of $(n!)^2$ is used to clear the factorial denominator of each
		non-prefix coordinate, the remaining sum over suffixes is exactly the
		completion counter $Z$.

		Fix an admissible candidate $t\in\{1,\ldots,B(p)\}$. For any valid suffix
		$A\subseteq E_t$, the full extension is $(p_1,\dots,p_r,t,a_1,\dots,a_L)$,
		where $a_1>\cdots>a_L$ are the elements of $A$. Writing $\Delta(p,t,A)$
		for the Vandermonde factor of this full extension, we have
		\[
		\Delta(p,t,A)^2
		=\Delta(p)^2\prod_{i=1}^r(p_i-t)^2\,
		\Delta(A)^2\prod_{a\in A}(t-a)^2
		\prod_{i=1}^r\prod_{a\in A}(p_i-a)^2 .
		\]
		The five factors on the right are, respectively, the prefix-only,
		prefix--candidate, suffix-only, candidate--suffix, and prefix--suffix
		interactions.
		Thus the full weight of this extension is
		\[
		K_p\,
		\frac{\prod_{i=1}^r(p_i-t)^2}{t!^2}\,
		\Delta(A)^2
		\prod_{a\in A}
		\frac{(t-a)^2\prod_{i=1}^r(p_i-a)^2}{a!^2},
		\]
		where $K_p:=\Delta(p)^2\prod_{i=1}^r p_i!^{-2}$ depends only on $p$.
		Multiplying everything except $K_p$ by $(n!)^{2(L+1)}$, one factor of
		$(n!)^2$ for each non-prefix coordinate, gives
		\[
		G_p(t)\,\Delta(A)^2\prod_{a\in A}\psi_{p,t}(a),
		\]
		which is $G_p(t)$ times the summand in the definition of
		$Z_{L,s(t)}(E_t,\psi_{p,t})$.
		Summing over all valid suffixes $A$, the total $\wt$-mass of all
		extensions with next coordinate $t$ is
		\[
		K_p(n!)^{-2(L+1)}G_p(t)Z_{L,s(t)}(E_t,\psi_{p,t})
		=K_p(n!)^{-2(L+1)}S_p(t).
		\]
		The proportionality constant is independent of $t$, so $S_p(t)$ is
		proportional to the required total mass.
	\end{proof}

	The scores at a stage aggregate into cumulative scores, which the fast
	implementation of Section~\ref{sec:complexity} queries directly. For a
	prefix $p$ of length $r$, set $R:=N-\sum_{i=1}^r p_i$ and, for
	$0\le T\le B(p)$,
	\[
	C_p(T):=Z_{L+1,R}\bigl(\{1,\dots,T\},G_p\bigr)
	=\sum_{\substack{A\subseteq\{1,\dots,T\}\\|A|=L+1,\ \sum A=R}}
	\Delta(A)^2\prod_{a\in A}G_p(a).
	\]
	In words, $C_p(T)$ is the total scaled mass of all completions whose
	remaining coordinates all lie in $\{1,\dots,T\}$.
	Since the remaining $L+1$ coordinates form a strictly decreasing sequence,
	the next coordinate is the largest element of the remaining set, which
	gives the following identity.

	\begin{lemma}[Cumulative scores]\label{lem:cum-mass}
		For every $1\le T\le B(p)$,
		\[
		C_p(T)-C_p(T-1)=S_p(T),
		\qquad\text{and hence}\qquad
		C_p(T)=\sum_{t\le T}S_p(t).
		\]
	\end{lemma}

	\begin{proof}
		The summands of the difference are exactly the sets $A$ with
		$\max A=T$, that is, $A=A'\cup\{T\}$ with $A'\subseteq\{1,\dots,T-1\}$,
		$|A'|=L$, and $\sum A'=R-T=s(T)$. For such $A$,
		\[
		\Delta(A)^2\prod_{a\in A}G_p(a)
		=G_p(T)\,\Delta(A')^2\prod_{a\in A'}(T-a)^2G_p(a)
		=G_p(T)\,\Delta(A')^2\prod_{a\in A'}\psi_{p,T}(a),
		\]
		and summing over $A'$ gives exactly
		$S_p(T)=G_p(T)\,Z_{L,s(T)}(E_T,\psi_{p,T})$.
	\end{proof}

	Consequently the weighted draw at a stage needs only $O(\log n)$
	cumulative scores. Compute $W:=C_p(B(p))=\sum_{t\le B(p)}S_p(t)$, draw a
	uniform integer $U\in\{1,\dots,W\}$, and binary-search the least $T$ with
	$C_p(T)\ge U$. Since $C_p$ is nondecreasing, the returned $T$ has
	probability exactly $S_p(T)/W$, the law of the weighted draw in
	Algorithm~\ref{alg:shape-sampler}.

	\subsection{A determinant oracle for completion counts}
	\label{sec:oracle}

	The completion score $S_p(t)$ involves the counter
	$Z_{L,s(t)}(E_t,\psi_{p,t})$, a sum of $\Delta(A)^2$-weighted terms over all
	$L$-element subsets $A\subseteq E_t$ with $\sum A=s(t)$. Enumerating these subsets directly is infeasible as there are
	\(\binom{t-1}{L}\) such subsets, which can be exponential in \(n\). The following
	lemma shows that $Z_{L,s}(E,\psi)$ is instead recoverable as a single
	coefficient of the determinant of an $L\times L$ polynomial matrix, and that this
	determinant can be computed in polynomial time.  This determinant
	representation is what makes the general sampler polynomial-time.

	Let $E=\{e_1<\cdots<e_m\}$ be a finite set of nonnegative integers, $R$ a
	commutative ring, and $\psi:E\to R$ a weight function. For $L\ge 0$, define the following $L\times L$ matrix over $R[q]$.
	\[
	M_{L,E,\psi}(q):=\sum_{y\in E}\psi(y)q^y\,v(y)v(y)^\top,
	\quad
	v(y):=(1,y,\dots,y^{L-1})^\top\in R^L.
	\]
	This is the weighted moment matrix of $E$, the sum of the rank-one terms
	$v(y)v(y)^\top$ scaled by $\psi(y)q^y$.  Each anti-diagonal $i+j=d$ holds a
	single weighted power sum $\sum_{y\in E}\psi(y)q^y y^{d}$, so $M_{L,E,\psi}(q)$
	is determined by only $2L-1$ such sums.  It is built so that its determinant
	generates the completion counts $Z_{L,s}$, as the next lemma shows, via the
	factorisation $M_{L,E,\psi}(q)=VD(q)V^\top$ and the Cauchy--Binet formula.

	\begin{lemma}[Determinant oracle]\label{lem:oracle}
		For every $L\ge 0$ and integer $s\ge 0$,
		\[
		[q^s]\det M_{L,E,\psi}(q)=Z_{L,s}(E,\psi).
		\]
	\end{lemma}

	\begin{proof}
		The case $L=0$ is immediate. For $L\ge 1$, write $M_{L,E,\psi}(q)=VD(q)V^\top$
		where $V$ is the $L\times m$ Vandermonde matrix with columns $v(e_j)$ and
		$D(q)=\operatorname{diag}(\psi(e_1)q^{e_1},\dots,\psi(e_m)q^{e_m})$. By the
		Cauchy--Binet formula~\cite[Section~0.8.7]{HornJohnson2013},
		\[
		\det\bigl(VD(q)V^\top\bigr)
		=\sum_{\substack{J\subseteq[m]\\|J|=L}}
		\det(V_J)\det(D_J(q))\det(V_J^\top),
		\]
		where $V_J$ and $D_J(q)$ are the submatrices indexed by $J$. Setting
		$A=\{e_j:j\in J\}$, we have $\det(V_J)=\Delta(A)$ and
		$\det(D_J(q))=\prod_{a\in A}\psi(a)q^a$, so
		\[
		\det M_{L,E,\psi}(q)
		=\sum_{\substack{A\subseteq E\\|A|=L}}
		\Delta(A)^2\prod_{a\in A}\psi(a)\,q^{\sum A}.
		\]
		Extracting the coefficient of $q^s$ gives $Z_{L,s}(E,\psi)$.
	\end{proof}

	The general ground set $E$ and weight $\psi$ are needed only for
	Lemma~\ref{lem:oracle}. In the sampler, the prefix $p$ determines $L$, and
	each candidate $t$ determines $E_t$ and $\psi_{p,t}$. We therefore
	abbreviate
	\[
	M_t(q):=M_{L,E_t,\psi_{p,t}}(q).
	\]
	By Lemma~\ref{lem:oracle}, the completion score is
	\[
	S_p(t)=G_p(t)\,[q^{s(t)}]\det M_t(q).
	\]
	Thus, for every prefix $p$ and every candidate next coordinate $t$, the total
	weight of all completions below $t$ can be computed exactly by a single
	determinant coefficient. Algorithmically, we extract this coefficient by evaluating the determinant
	polynomial at sufficiently many field elements and interpolating.

	\subsection{The sequential shape sampler}
	\label{sec:shape-sampler}

	We now assemble the one-coordinate-at-a-time sampler
	(Algorithm~\ref{alg:shape-sampler}). Start with the empty
	prefix $p=\varnothing$. At each stage, compute the completion scores $S_p(t)$ for
	all admissible $t$, then draw the next coordinate with probability
	\[
	\Pp(x_{r+1}=t\mid p)
	=
	\frac{S_p(t)}{\sum_{u=1}^{B(p)}S_p(u)}.
	\]

	\begin{algorithm}[h]
		\caption{Sequential shape sampler}
		\label{alg:shape-sampler}
		\begin{algorithmic}[1]
			\Require Integers $n,k$ with $1\le k\le n$
			\State Set $N\gets n+\binom{k}{2}$.
			\State Set $p\gets\varnothing$.
			\For{$r=0,1,\ldots,k-1$}
			\State Set $B\gets n$ if $r=0$, and $B\gets p_r-1$ otherwise.
			\For{each $t\in\{1,\ldots,B\}$}
			\State Compute the completion score $S_p(t)$ via Lemma~\ref{lem:oracle},
			where $s(t)=N-\sum_{i=1}^{r}p_i-t$.
			\EndFor
			\State Sample $t\in\{1,\ldots,B\}$ with probability proportional to the weights $S_p(t)$.
			\State Append $t$ to the prefix $p$.
			\EndFor
			\State \Return $x=p$.
		\end{algorithmic}
	\end{algorithm}

	Algorithm~\ref{alg:shape-sampler} is stated in its direct form, computing
	every score at every stage. It defines the sampling law. The fast
	implementation of Section~\ref{sec:complexity} realises the same weighted
	draw through the cumulative scores of Lemma~\ref{lem:cum-mass}, without
	computing individual scores.

	For an admissible candidate \(t\), recall from Section~\ref{sec:scores}
	that \(s(t)=N-\sum_{i=1}^r p_i-t\).
	A suffix of length \(L\) chosen from \(E_t=\{1,\ldots,t-1\}\) can have total
	sum only between the sum of the \(L\) smallest positive integers and the sum of
	the \(L\) largest elements below \(t\). Hence \(S_p(t)=0\) unless
	\[
	\binom{L+1}{2}\le s(t)\le Lt-\binom{L+1}{2}.
	\]
	Candidates failing this test may therefore be skipped before invoking the
	determinant oracle. After $k$ coordinates have been
	chosen, we recover $\rho_i=x_i-k+i$ and set $\lambda=\rho'$.

	The weighted draw at each stage is well defined.  The initial state is
	feasible, since for every $1\le k\le n$, the shape $\lambda=(k,1^{n-k})$
	corresponds to the strict partition $x=(n,k-1,k-2,\ldots,1)$, which has $k$
	parts, maximum $n$, and sum $N$, so $\sum_t S_\varnothing(t)>0$.
	Inductively, every reachable prefix was extended by a coordinate drawn with
	positive completion score, so it admits at least one valid completion and
	the total score at the next stage is again positive.

	For a prefix $p$, let $W(p)$ denote the total $\wt$-mass of all valid full
	strict partitions extending $p$.

	\begin{theorem}\label{thm:shape-correct}
		The sequential sampler outputs a strict partition $x$ of $N$ with $k$ parts
		with probability proportional to $\wt(x)$. Equivalently, it samples
		$\lambda\vdash n$ with $\lambda_1=k$ from the exact law
		$\Pp_{n,k}(\lambda)\propto(f^\lambda)^2\,\mathbf 1_{\{\lambda_1=k\}}$.
	\end{theorem}

	\begin{proof}
		The probability of producing a complete \(x\) is the product of its
		conditional coordinate probabilities. By Lemma~\ref{lem:score-mass},
		each conditional probability is
		\(W(x_1,\ldots,x_{r+1})/W(x_1,\ldots,x_r)\), so these ratios telescope:
		\[
		\prod_{r=0}^{k-1}\frac{W(x_1,\dots,x_{r+1})}{W(x_1,\dots,x_r)}
		=\frac{\wt(x)}{W(\varnothing)},
		\]
		where $W(x)=\wt(x)$ since a complete strict partition has no further
		suffix, and $W(\varnothing)$, the total mass of all valid strict
		partitions, is the normalising constant. Hence the
		output law is proportional to $\wt(x)$, and Lemma~\ref{lem:x-hook} identifies
		this with the target shape distribution.
	\end{proof}

	\subsection{From sampled shapes to sampled permutations}\label{sec:tableaux}

	The remaining steps are standard. We use the Robinson--Schensted bijection
	between permutations of \([n]\) and pairs \((P,Q)\) of standard Young tableaux
	of the same shape, together with its inverse reverse-insertion algorithm
	\cite[Theorem~3.1.1 and proof, p.~94]{Sagan2001}. We also use the
	Greene--Nijenhuis--Wilf hook walk, which samples a uniform standard Young
	tableau of any fixed shape \cite{GreeneNijenhuisWilf1979}.

	\begin{lemma}\label{lem:shape-to-perm}
	Suppose that $\lambda\vdash n$ is sampled from
	$\Pp_{n,k}(\lambda)\propto (f^\lambda)^2\,\mathbf 1_{\{\lambda_1=k\}}$.
	Conditional on $\lambda$, sample $P$ and $Q$ independently and uniformly from $\SYT(\lambda)$, and output the inverse Robinson--Schensted image of $(P,Q)$. Then the output is a uniformly random permutation in $\Omega_{n,k}$.
	\end{lemma}

	\begin{proof}
	For fixed $\lambda$, Robinson--Schensted identifies the permutations of shape $\lambda$ with $\SYT(\lambda)\times\SYT(\lambda)$. Hence the procedure gives each permutation of shape $\lambda$ probability
	$\Pp_{n,k}(\lambda)\cdot (f^\lambda)^{-2}$,
	which is independent of $\lambda$ on the support $\{\lambda\vdash n:\lambda_1=k\}$. Since Robinson--Schensted sends the LIS length to the first row length \cite{Schensted1961}, this support is exactly $\Omega_{n,k}$.
	\end{proof}

	Both steps cost $O(n^2)$ time. The two Greene--Nijenhuis--Wilf hook walks
	sample $P$ and $Q$ in $O(n^2)$ time. For the inverse Robinson--Schensted step,
	precompute the positions of the labels in the recording tableau. Each
	reverse insertion then traces one bumping path upward through the
	insertion tableau. Implemented by linear scans of the visited rows, it
	examines at most the current number of tableau entries, so the $n$
	reverse insertions cost $O(n^2)$ in total.

	\subsection{Running time of the direct implementation}\label{sec:direct-time}

Throughout the running-time analysis (this section and
Section~\ref{sec:complexity}) we assume $k\ge2$. For $k=1$ the only shape is
$\lambda=(1^n)$, so the sampler is trivial and every stage has $L=0$, invoking no
oracle.
We work in the model of Section~\ref{sec:results} and evaluate the oracle of
Section~\ref{sec:oracle} separately for each candidate. Fix a prefix
$p=(p_1,\dots,p_r)$, write $L:=k-r-1$, and recall
$G_p(a)=(n!/a!)^2\prod_{i=1}^r(p_i-a)^2$ and $\psi_{p,t}(a)=(t-a)^2G_p(a)$. The
residues $\{G_p(a)\bmod\mathfrak p\}_{a\le n}$ are rebuilt for each
sequentially processed prime and stage in $O(nk)$ field operations by one
pass over $a$, a lower-order cost.

\paragraph{Bit lengths and modular arithmetic.}
\begin{lemma}\label{lem:bits-stage}
For every stage with prefix $p$ and suffix length $L$, and every feasible
candidate $t$,
the integer completion score $S_p(t)$ has $O((L+1)n\log n)=O(kn\log n)$ bits.
\end{lemma}
\begin{proof}
We bound the factorial factors and the difference factors of a single
summand separately, then account for the number of summands.
Every nonzero summand of $S_p(t)$ is
$G_p(t)\,\Delta(A)^2\prod_{a\in A}(t-a)^2G_p(a)$ for an $L$-element
$A\subseteq[t-1]$. The $L+1$ factorial factors contribute at most
$(n!)^{2(L+1)}$, i.e.\ $O((L+1)n\log n)$ bits. The $\binom L2+rL+r+L=O(nL+n)$
difference factors (using $r+L+1=k\le n$), each at most $n$, contribute
$O((L+1)n\log n)$ more. The at most $\binom nL$ terms add $O(L\log n)$. The
bound follows.
\end{proof}

Fix a constant $C$ so that every score and every sum of at most $n$
scores has at most $\beta:=Ckn\log n$ bits, and set
$d_{\max}:=nk+1\le n^2+1$, a bound on the evaluation points used anywhere.
Take distinct $\Theta(\log n)$-bit primes $\mathcal P$, each exceeding
$d_{\max}$, with $\prod_{\mathfrak p\in\mathcal P}\mathfrak p>2^\beta$. Then
$|\mathcal P|=\tilde O(kn)$ suffice, found by a segmented sieve over
$(d_{\max},n^2\log^2n]$ in $\tilde O(n^2)$ running time using $\tilde O(n)$
words of workspace\footnote{A Monte Carlo primality test
would not do as an undetected composite modulus could corrupt the output law.}.
At any stage that needs $d\le d_{\max}$ evaluation points we use the field
elements $\xi_u=u$ for $u=0,1,\ldots,d-1$, which are distinct in every
$\mathbb F_{\mathfrak p}\in\mathcal P$ because $\mathfrak p>d_{\max}\ge d$.

We use the following \textbf{CRT accounting convention}. First analyse the oracle over
one word-size prime field $\mathbb F_{\mathfrak p}$. To obtain exact integer
completion scores, repeat the same computation modulo enough such primes that
their product exceeds the a priori upper bound $2^\beta$ on every score, and then
reconstruct the scores by the Chinese remainder theorem. Since $\beta=O(kn\log n)$ and the primes have $\Theta(\log n)$ bits, obtaining exact integer scores from one-prime computations multiplies the running time by $\tilde O(kn)$ in the word-RAM model.
At a stage with suffix length $L$, Lemma~\ref{lem:bits-stage} bounds every
score, and hence every sum of at most $n$ of them, by $O((L+1)n\log n)$
bits, so the first $\tilde O((L+1)n)$ primes of $\mathcal P$ already
suffice there. The fast implementation of Section~\ref{sec:complexity}
uses this per-stage refinement.
Since each prime is word-size, a single $\mathbb F_{\mathfrak p}$-operation costs
$O(1)$ word operations ($O(\log n)$ for inversion), so a per-prime
field-operation count and the corresponding word-operation count agree up to the
polylogarithmic factors hidden in $\tilde O$.
The reconstruction is exact, so the weighted choices use the true integer
scores, not approximations.

\paragraph{Per-stage cost and the overall running time.}
\begin{lemma}\label{lem:direct-stage-oracle}
If the oracle of Lemma~\ref{lem:oracle} is evaluated separately for each
candidate, all completion scores at a stage with suffix length $L\ge1$ are
computed in $\tilde O(n^4kL^3)$ word operations.
\end{lemma}
\begin{proof}
Work over a prime field $\mathbb F_{\mathfrak p}$ with
$\mathfrak p>d_{\max}$. Since $\mathbb F_{\mathfrak p}$ has
$\mathfrak p$ elements, this gives at least $d_{\max}+1$ distinct
evaluation points, enough to interpolate any polynomial of degree at most
$d_{\max}$. Fix a candidate $t$. The
polynomial $D_t(q):=\det M_t(q)$ has degree at most $Ln$, so we
evaluate it at the points $\xi_0,\dots,\xi_{d-1}$ with $d:=Ln+1$ and
interpolate. At one point, forming
$M_t(\xi_u)$ from its defining sum costs $O(nL^2)$ field
operations, and computing its determinant by Gaussian elimination costs
$O(L^3)\subseteq O(nL^2)$ field operations. Hence one candidate costs $O(d\cdot nL^2)=O(n^2L^3)$,
including the $O(d^2)=O(n^2L^2)$ cost of interpolation. Over at most $n$ candidates this is $O(n^3L^3)$ field
operations per prime, hence $\tilde O(n^4kL^3)$ word operations by the CRT
convention.
\end{proof}

We are now in a position to prove Theorem~\ref{thm:intro-general-direct}, the
direct implementation's $\tilde O(n^4k^5)$ running time bound.

\begin{proof}[Proof of Theorem~\ref{thm:intro-general-direct}]
A run has $k$ stages, of suffix lengths $k-1,\dots,1,0$, the terminal $L=0$ stage
being trivial. Summing Lemma~\ref{lem:direct-stage-oracle} over the nontrivial
stages, $\sum_{L=1}^{k-1}\tilde O(n^4kL^3)=\tilde O(n^4k^5)$. The exact weighted
draws use only lower-order arithmetic on $O(kn\log n)$-bit integers. Given the sampled shape, the tableau-sampling and inverse
Robinson--Schensted steps cost $O(n^2)$
(Section~\ref{sec:tableaux}), which is absorbed. By Lemma~\ref{lem:shape-to-perm} the output is uniform on
$\Omega_{n,k}$, with expected running time $\tilde O(n^4k^5)$.
\end{proof}

	\section{Running time of the fast implementation}\label{sec:complexity}

The direct implementation of Section~\ref{sec:direct-time} is wasteful in two
ways. It computes up to $n$ candidate scores at every stage, and for every
candidate and every interpolation point it rebuilds the evaluated matrix and
takes its determinant as a generic dense matrix. The fast implementation
removes both costs, leaving the output law untouched. It draws each coordinate
through the cumulative scores of Lemma~\ref{lem:cum-mass}, so that a stage
needs only $O(\log n)$ oracle queries rather than $n$ scores, and it evaluates
each query through the Hankel structure of the evaluated matrices. The stage
weights $G_p$, the prime set $\mathcal P$, the bit-length bound
(Lemma~\ref{lem:bits-stage}), and the CRT accounting convention of
Section~\ref{sec:direct-time} are reused. Only the realisation of the weighted
draw changes, improving the bound of Theorem~\ref{thm:intro-general-direct} to
that of Theorem~\ref{thm:intro-general-fast}.

\paragraph{From one determinant to the full faster sampler.}
The running time bound is built from the bottom up: a single determinant, then one cumulative-score query, then one sampling stage, then all stages. Each cumulative score is a single coefficient of a polynomial determinant. We recover that coefficient by evaluating the determinant at sufficiently many scalar points and interpolating, working modulo word-size primes and reconstructing the exact integer scores by the CRT. The key structural fact, recorded in the next subsection, is that once the determinant polynomial is evaluated at a scalar the resulting matrix is Hankel.
Lemma~\ref{lem:hankel-primitive} exploits this to compute a single Hankel
determinant in $\tilde O(L)$ field operations, the superfast
primitive that replaces dense elimination. Building on it, Lemma~\ref{lem:fast-hankel-stage} evaluates one cumulative score by batched multipoint evaluation of the Hankel moments, in $\tilde O((L+1)^2T)$ field operations per prime. Lemma~\ref{lem:fast-shape-time} then implements each stage by a binary search over $O(\log n)$ such queries and sums over the stages, bounding the expected running time of the shape sampler by $\tilde O(n^2k^4)$. The
proof of Theorem~\ref{thm:intro-general-fast} introduces the tableau-sampling
and inverse Robinson--Schensted steps while verifying that the output law is
unchanged.

\subsection{Hankel structure of the evaluated matrices}

Fix a prefix $p$ of length $r$ and recall $L=k-r-1$, the number of
	coordinates that remain to be chosen. Two parameters recur throughout this
	subsection. The query point $T$ is the cutoff in the cumulative score $C_p(T)$
	requested by the binary search, and the evaluation point $\xi$ is a field
	element substituted for the variable $q$ in the determinant polynomial.
	By Lemma~\ref{lem:oracle}, applied with ground set
	$\{1,\dots,T\}$ and weight $G_p$,
	\[
	C_p(T)=[q^{R}]\det M_{L+1,\{1,\dots,T\},G_p}(q),
	\]
	and at the scalar $\xi$ the matrix becomes
	\[
	M_{L+1,\{1,\dots,T\},G_p}(\xi)
	=
	\sum_{a\le T}G_p(a)\xi^a\,v(a)v(a)^\top,
	\qquad v(a)=(1,a,\dots,a^{L})^\top.
	\]
	Using zero-based indices $0\le i,j\le L$, its $(i,j)$ entry is
	$\sum_{a\le T}G_p(a)\xi^a a^{i+j}$, which depends only on $i+j$, so the
	evaluated matrix is Hankel, being determined by the $2L+1$ anti-diagonal
	quantities for $i+j=0,1,\ldots,2L$, which we call moments.

	\begin{lemma}[Superfast Hankel determinants]
		\label{lem:hankel-primitive}
		Let $F$ be a field. Given elements $h_0,\ldots,h_{2L-2}\in F$, form the
		$L\times L$ Hankel matrix $H=(h_{i+j})_{0\le i,j<L}$.
		Then $\det H$ can be computed in $\tilde O(L)$ arithmetic
		operations in $F$, without assuming that $H$ is nonsingular or that any leading principal Hankel minor is nonzero.
	\end{lemma}

\begin{proof}
	Let
	\[
	h(z)=h_0+h_1z+\cdots+h_{2L-2}z^{2L-2}.
	\]
	For $j\le L$, the order-$j$ Hankel determinant
	\[
	\mathcal H_j(h):=\det(h_{a+b})_{0\le a,b<j}
	\]
	uses only the coefficients $h_0,\ldots,h_{2j-2}$, since the largest index
	appearing in the matrix is $2j-2$. In particular,
	$\mathcal H_1(h),\ldots,\mathcal H_L(h)$ are completely determined by the $2L-1$ coefficients
	$h_0,\ldots,h_{2L-2}$. 
	
	We invoke the algorithm of Liu, Xin and Zhang \cite{liu2025n}. First we note that their
	main result is not stated in exactly the form we need. Its input is a generating function or
	\emph{rational} power series, that is, a power series expressible as a ratio
	$N(z)/D(z)$ of two polynomials with $D(0)\neq 0$, taken in lowest terms
	($\gcd(N,D)=1$); it is phrased over $\mathbb C$; and it returns an entire
	initial determinant sequence rather than one determinant. We reconcile these
	three points in turn. None requires anything beyond what their proof already
	establishes.

	\emph{A polynomial is already a rational series.}\enspace We feed the
	polynomial $h(z)$ above, that is, the rational power series $N/D$ with $N=h$ and
	$D=1$. The coprimality hypothesis $\gcd(N,D)=1$ then holds
	trivially, since $D=1$. As $D=1$ the series is genuinely the
	polynomial $h$, so exactly the $2L-1$ prescribed entries $h_0,\ldots,h_{2L-2}$
	enter the computation and no further coefficients are implicitly assumed.

	\emph{Only field arithmetic is used.}\enspace Their main algorithm for computing Hankel determinants (their Algorithm~3.5) and the results it rests on use only field arithmetic and
	Euclidean division of polynomials, organised through the half-GCD. They invoke
	no property of $\mathbb C$ beyond that it is a field\footnote{Their Section~3.2, on
	signatures and real-rootedness, does specialise to $\mathbb R$, but it is no
	part of the determinant computation and we do not use it.}. Hence the algorithm,
	and its $O(L\log^2 L)$ field-operation bound, apply verbatim over any field
	$F$, in particular over the prime fields $\mathbb F_{\mathfrak p}$ over which
	our oracle is evaluated (Section~\ref{sec:direct-time}).

	\emph{One determinant, with vanishing minors allowed.}\enspace If
	$h_0=\cdots=h_{2L-2}=0$ then $h=0$, $H$ is the zero matrix, and $\det H=0$. So,
	assume $h\neq0$ and write $\delta:=\deg h\le 2L-2$. With $D=1$, their Theorem~3.6 truncates the Hankel
	determinant sequence at order $d=\max(\deg D,\deg N+1)$, which here is $\delta+1$. Run their algorithm with parameter $L$. It returns the order-$j$ Hankel
	determinants $\mathcal H_1(h),\mathcal H_2(h),\ldots,\mathcal H_{\min(L,d)}(h)$, from order $1$ up to the order
	$\min(L,d)$ at which the output is clamped. Two cases arise according to this
	upper limit. If $L\le d$, the clamp is inactive and the last returned term is
	exactly
	\[
	\mathcal H_L(h)=\det(h_{i+j})_{0\le i,j<L}=\det H,
	\]
	as required. If instead $L>d$, then every entry in the last column of $H$
	has index at least $L-1>\delta$. The last column therefore vanishes, so
	$\det H=\mathcal H_L(h)=0$. Either way $\det H$ is obtained.
\end{proof}

Classical superfast Toeplitz methods go back to
\cite{BrentGustavsonYun1980}, and singular Toeplitz-like and Hankel-like
algorithms are treated in
\cite{Pan2000StructuredMatrices,PanZhengAbuTabanjehChenProvidence1999}.
Those methods could also be used after converting \(H\) to a Toeplitz matrix.
The singular case then proceeds through randomised generic-rank-profile
preconditioning and certification. We use \cite{liu2025n} because it applies
directly to the Hankel moment sequence and avoids this extra bookkeeping.

	We now bound the cost of one cumulative-score query.

	\begin{lemma}[Cumulative-score oracle]
		\label{lem:fast-hankel-stage}
		Fix a stage with prefix \(p\) and suffix length \(L\ge 0\), and a query
		point \(1\le T\le B(p)\). Given the residues
		\(\{G_p(a)\bmod\mathfrak p:1\le a\le n\}\), the residue of \(C_p(T)\)
		modulo one prime \(\mathfrak p\in\mathcal P\) is computable in
		\(\tilde O((L+1)^2T)\) field operations, and the exact integer
		\(C_p(T)\) in \(\tilde O((L+1)^3nT)\) word operations.
	\end{lemma}

	\begin{proof}
		If \(T<L+1\) there is no \((L+1)\)-element subset of
		\(\{1,\dots,T\}\) and \(C_p(T)=0\). Assume \(T\ge L+1\) and write
		\(\widehat D_T(q):=\det M_{L+1,\{1,\dots,T\},G_p}(q)\). Every monomial
		\(q^{\sum A}\) contributing to \(\widehat D_T\) has
		\(\sum A\le T+(T-1)+\cdots+(T-L)=(L+1)T-\binom{L+1}2\), so
		\[
		d:=(L+1)T-\binom{L+1}2+1=O\bigl((L+1)T\bigr)\le nk+1=d_{\max}
		\]
		evaluation points \(\xi_0,\dots,\xi_{d-1}\) determine \(\widehat D_T\),
		and they are distinct in \(\mathbb F_{\mathfrak p}\) since
		\(\mathfrak p>d_{\max}\).

		\emph{Batched moment evaluation.}\enspace Define the \(2L+1\) moment
		polynomials
		\[
		h_\ell(q):=\sum_{a=1}^T a^\ell G_p(a)\,q^a,
		\qquad 0\le\ell\le 2L,
		\]
		so that \(M_{L+1,\{1,\dots,T\},G_p}(\xi)=(h_{i+j}(\xi))_{0\le i,j\le L}\).
		Their coefficient arrays are built in \(O((L+1)T)\) field operations,
		keeping the powers \(a^0,\dots,a^{2L}\) incrementally. Fast multipoint
		evaluation of one degree-\(T\) polynomial at the \(d\) points costs
		\(\tilde O(T+d)=\tilde O((L+1)T)\) field operations
		\cite[Chapter~10]{vonZurGathenGerhard2013}, so evaluating all
		\(2L+1\) moments at all \(d\) points costs \(\tilde O((L+1)^2T)\).

		\emph{Determinants and interpolation.}\enspace At each point \(\xi_u\)
		the evaluated matrix is Hankel with moments
		\(h_0(\xi_u),\dots,h_{2L}(\xi_u)\), so
		Lemma~\ref{lem:hankel-primitive} computes \(\widehat D_T(\xi_u)\) in
		\(\tilde O(L+1)\) field operations, another \(\tilde O((L+1)^2T)\) in
		total. By the same standard algorithms, fast interpolation recovers the
		coefficients of \(\widehat D_T\) from the \(d\) values in
		\(\tilde O(d)=\tilde O((L+1)T)\) field operations, and
		\(C_p(T)=[q^{R}]\widehat D_T(q)\) is read off, being zero when
		\(R\ge d\).

		One prime therefore costs \(\tilde O((L+1)^2T)\) field operations. By
		Lemma~\ref{lem:bits-stage}, \(C_p(T)\), a sum of at most \(n\) scores,
		has \(O((L+1)n\log n)\) bits, so the per-stage refinement of the CRT
		convention reconstructs the exact integer from the first
		\(\tilde O((L+1)n)\) primes of \(\mathcal P\), in
		\(\tilde O((L+1)^2T)\cdot\tilde O((L+1)n)=\tilde O((L+1)^3nT)\) word
		operations.
	\end{proof}

	\begin{lemma}[Fast shape-sampling time]\label{lem:fast-shape-time}
		The exact shape sampler has expected running time \(\tilde O(n^2k^4)\).
	\end{lemma}

	\begin{proof}
		Consider a stage with prefix \(p\) and suffix length \(L\). As
		described after Lemma~\ref{lem:cum-mass}, the stage computes
		\(W=C_p(B(p))\), draws a uniform \(U\in\{1,\dots,W\}\), and
		binary-searches the least \(T\le B(p)\) with \(C_p(T)\ge U\), which
		realises the weighted draw of Algorithm~\ref{alg:shape-sampler}
		exactly. This takes \(O(\log n)\) queries, each costing
		\(\tilde O((L+1)^3nT)\subseteq\tilde O(n^2(L+1)^3)\) word operations by
		Lemma~\ref{lem:fast-hankel-stage} and \(T\le B(p)\le n\), plus lower-order
		comparisons of \(O((L+1)n\log n)\)-bit integers. Preparing the
		residues \(\{G_p(a)\bmod\mathfrak p\}\) across the
		\(\tilde O((L+1)n)\) primes of the stage costs
		\(O(nk)\) field operations per prime, hence
		\(\tilde O((L+1)kn^2)\) word operations per stage. Summing over the
		stages \(L=k-1,\dots,0\),
		\[
		\sum_{L=0}^{k-1}\tilde O\bigl(n^2(L+1)^3+(L+1)kn^2\bigr)
		=\tilde O(n^2k^4).\qedhere
		\]
	\end{proof}

	We now prove Theorem~\ref{thm:intro-general-fast}, the fast
	implementation's $\tilde O(n^2k^4)$ running time bound.

	\begin{proof}[Proof of Theorem~\ref{thm:intro-general-fast}]
		The sampling law is the same as in Theorem~\ref{thm:intro-general-direct}. Only the realisation of the weighted draw changes. By
		Lemma~\ref{lem:fast-shape-time}, the shape distribution is sampled in expected
		\(\tilde O(n^2k^4)\) word operations. The subsequent tableau-sampling choices
		have expected \(O(n^2)\) running time, and the inverse Robinson--Schensted
		step is deterministic and takes \(O(n^2)\) word operations. Exactness is
		unchanged, since by Lemma~\ref{lem:cum-mass} the binary search over
		exact integer cumulative scores realises the same weighted draw as
		Algorithm~\ref{alg:shape-sampler}.
	\end{proof}

	\smallskip
	\noindent\textbf{Dispensing with the superfast primitive.}\enspace
	The superfast primitive can be dispensed with at a polynomial cost.
	Taking each evaluated Hankel determinant by plain Gaussian elimination
	($O((L+1)^3)$) instead of Lemma~\ref{lem:hankel-primitive} makes a query
	cost $\tilde O((L+1)^4T)$ field operations per prime, giving
	$\tilde O(n^2k^6)$ word operations for the full sampler. This is never asymptotically larger than the
	direct bound, while using no structured-matrix primitive.
	\smallskip

	\noindent\textbf{Space usage.}\enspace
		The peak space usage of the fast implementation is $\tilde O(kn)$
		words. Consider one query. Per prime, the moment polynomials occupy
		$O((L+1)T)$ field elements, and the $d$ evaluation points are
		processed in blocks of $\Theta(T)$. On each block, the $2L+1$ moments
		are evaluated, the $\Theta(T)$ determinants are computed, and the
		moment values are discarded, retaining only the $O(d)=O((L+1)T)$
		determinant values needed for interpolation. The workspace per prime
		is therefore $\tilde O((L+1)T)\le\tilde O(kn)$ field elements, reused
		across the sequentially processed primes. The CRT computation retains
		one residue for each of the $\tilde O((L+1)n)$ primes used at the
		stage, and the reconstructed integers $W$ and $C_p(T)$ and the drawn
		integer $U$ each occupy $O((L+1)n\log n)$ bits. The peak is
		$\tilde O(kn)$ words in all. The direct implementation satisfies the
		same bound. It computes candidate scores one at a time, retaining only
		the interpolation data for the current candidate. A first pass computes
		the total score, and after drawing the target integer, a second pass
		locates it by accumulating the candidate scores.
	\smallskip

	\section{A direct sampler when \texorpdfstring{$k \in \Theta(n)$}{k in Theta(n)}}
	\label{sec:direct}

	In the linear regime $k\in\Theta(n)$ a much simpler and near-optimal sampler is
	available.  It works directly with permutations and uses rejection sampling on
	an expanded space of pairs $(\pi,I)$, where $I$ is a $k$-term increasing
	subsequence of $\pi$.  The acceptance rule is designed so that, for each
	permutation with $\LIS=k$, exactly one associated pair is accepted, namely the pair
	for which $I$ is the leftmost LIS.  The unique accepted witness ensures
	uniformity, and for $k\in\Theta(n)$ the acceptance probability turns out
	to be bounded below by a positive constant.  We summarise the construction here. The
	full analysis is in Section~\ref{sec:linear}.  Throughout let $m:=n-k$.  The
	construction and its exactness are valid for every $1\le k\le n$. The constant
	acceptance-probability bound, and hence the time bound, is for $k\in\Theta(n)$.

	Define the expanded proposal space
	\[
	\mathcal E_{n,k}:=
	\{(\pi,I): \pi\in S_n,\;
	I=(i_1<\cdots<i_k),\;
	\pi(i_1)<\cdots<\pi(i_k)\},
	\]
	so $I$ is a list of positions whose entries form an increasing subsequence of
	$\pi$ of length $k$.  We sample $(\pi,I)\in\mathcal E_{n,k}$ as follows.
	Choose a $k$-subset $V\subseteq[n]$ of values and a $k$-subset $I\subseteq[n]$
	of positions, both uniformly. Place the values of $V$ in increasing order into
	the positions $I$ and fill the remaining $m$ positions with the remaining $m$
	values in uniformly random order.

	\begin{proposition}\label{prop:proposal}
		The construction above is uniform on $\mathcal E_{n,k}$. In particular,
		$|\mathcal E_{n,k}|=\binom{n}{k}^{2}m!$.
	\end{proposition}

	For a permutation $\pi$ with $\LIS(\pi)=k$, let $L^\star(\pi)$ denote the
	leftmost LIS of $\pi$, that is, among all increasing subsequences of length $k$, the one
	whose list of positions is lexicographically smallest.  This is a canonical
	choice of LIS, computable in $O(n\log\log n)$ time
	(Section~\ref{sec:linear}).

	\begin{algorithm}[h]
		\caption{Direct rejection sampler}
		\label{alg:direct-sampler}
		\begin{algorithmic}[1]
			\Require Integers $n,k$ with $1 \le k \le n$
			\Loop
				\State Sample $(\pi,I)$ uniformly from $\mathcal E_{n,k}$.
				\State Compute \(\ell\gets \LIS(\pi)\).
				\If{\(\ell=k\)}
					\State Compute $L^\star(\pi)$.
					\If{$I=L^\star(\pi)$}
						\State \Return $\pi$.
					\EndIf
				\EndIf
			\EndLoop
		\end{algorithmic}
	\end{algorithm}

	\begin{proposition}\label{prop:direct-exact}
		Conditioned on acceptance in one iteration of Algorithm~\ref{alg:direct-sampler}, the output permutation is uniform on $\Omega_{n,k}$.
	\end{proposition}
	\begin{proof}
	The accepted proposals are precisely the pairs
	$\mathcal A_{n,k}:=\{(\pi,I)\in\mathcal E_{n,k}:\LIS(\pi)=k\text{ and } I=L^\star(\pi)\}$.
	The projection $(\pi,I)\mapsto\pi$ is a bijection from
	$\mathcal A_{n,k}$ to $\Omega_{n,k}$, since for each
	$\pi\in\Omega_{n,k}$ there is exactly one accepted pair, namely
	$(\pi,L^\star(\pi))$.  By Proposition~\ref{prop:proposal}, the proposal
	distribution is uniform on $\mathcal E_{n,k}$, so all accepted pairs have
	equal probability.  Since accepted pairs are in bijection with
	$\Omega_{n,k}$, the output permutation is uniform on $\Omega_{n,k}$.
	\end{proof}

	The acceptance probability of one iteration equals
	$|\Omega_{n,k}|/|\mathcal E_{n,k}|=A_{n,k}$, where $A_{n,k}$ is the
	Plancherel expectation of a squared \emph{hook-elongation factor}, defined
	in Section~\ref{sec:linear}. That factor records the loss in $f^{(k,\mu)}$
	incurred when a partition $\mu$ of $m$ is placed beneath a first row of
	length $k$. In Section~\ref{sec:linear} we show that $A_{n,k}$ is bounded
	below by a positive constant when $k\in\Theta(n)$. With the
	$O(n\log\log n)$ cost of generating and testing each proposal, this gives
	Theorem~\ref{thm:intro-largek}.

	The rejection sampler is not limited to the linear regime.  Its expected
	running time degrades gracefully as $k$ decreases, remaining polynomial as long as $k \in \Omega(n/\log n)$.

	\begin{corollary}[Below the linear regime]\label{cor:sublinear-k}
		For every \(k>4\sqrt m\), the expected running time of the rejection
		sampler of Theorem~\ref{thm:intro-largek} is
		\[
		O\!\left(n\exp\!\left(\frac{2m}{k-4\sqrt m+1}\right)\log\log n\right),
		\]
		where $m=n-k$.
		In particular, for every fixed \(\eta>0\) and all sufficiently large $n$, if \(k\ge \eta n/\log n\) then the
		expected running time is \(O\!\left(n^{1+2/\eta}\log\log n\right)\), which for
		\(\eta=1\) is \(O(n^3\log\log n)\).
	\end{corollary}
	This follows from a lower bound on the acceptance probability $A_{n,k}$, proved
	in Section~\ref{sec:linear}.

	\section{Full analysis of the linear-regime sampler}\label{sec:linear}

	This section proves Proposition~\ref{prop:proposal},
	Theorem~\ref{thm:intro-largek}, and the bound used in
	Corollary~\ref{cor:sublinear-k}.

	\subsection{Uniformity of the proposal}

	\begin{proof}[Proof of Proposition~\ref{prop:proposal}]
	The choices made by the proposal are exactly a $k$-set of values $V$, a
	$k$-set of positions $I$, and an ordering of the remaining $m=n-k$ values.
	Given these data, the entries on $I$ are forced to be the elements of $V$ in
	increasing order, and the remaining entries are filled according to the chosen
	ordering.

	Conversely, given $(\pi,I)\in\mathcal E_{n,k}$, the set $V$ is uniquely the
	set of values $\{\pi(i):i\in I\}$, and the ordering of the remaining values is
	read off from $\pi$ outside $I$. Thus the proposal data are in bijection with
	$\mathcal E_{n,k}$. Since the proposal data are chosen uniformly, the induced
	law on $\mathcal E_{n,k}$ is uniform. The cardinality formula follows from
	the same bijection.
	\end{proof}

	\subsection{Testing one proposal in \texorpdfstring{$O(n\log \log n)$}{O(n log log n)} time}
	For a permutation $\pi\in S_n$ and an index $i\in[n]$, let $r_i$ be the length
	of the longest increasing subsequence of $\pi$ that starts at position $i$.

	\begin{fact}\label{fact:lis-ending-ranks}
		Let $a\in S_n$. In the word-RAM model, one can compute, for every
		$j\in[n]$, the length $d_j$ of the longest increasing subsequence of $a$
		ending at position $j$ in $O(n\log\log n)$ time.
	\end{fact}

	This is the standard patience-sorting dynamic program with the ordered set of
	current tail values implemented by an integer predecessor structure, such as a
	van Emde Boas structure. See, e.g., Crochemore and
	Porat~\cite{CrochemorePorat2010}.

	\begin{lemma}\label{lem:ri}
		The values $r_1,\dots,r_n$ can be computed in $O(n\log \log n)$ time in the
		word-RAM.
	\end{lemma}

	\begin{proof}
		Let $b_j=n+1-\pi_{n+1-j}$ be the reversed-complement permutation.
		An increasing subsequence of $\pi$ starting at position $i$ corresponds
		bijectively to an increasing subsequence of $b$ ending at position $n+1-i$.
		Hence $r_i=d_{n+1-i}$, where $d_j$ is the length of the longest increasing
		subsequence of $b$ ending at $j$. By Fact~\ref{fact:lis-ending-ranks}, all
		$d_j$ are computed in $O(n\log\log n)$ time, and reading them in reverse
		gives all $r_i$ in the same time.
	\end{proof}

	Assume now that $\LIS(\pi)=k$. We construct the leftmost LIS greedily.
	Starting with the dummy values $i_0:=0$ and $v_0:=0$, suppose that
	$i_1,\dots,i_{t-1}$ have already been chosen and write
	$v_{t-1}:=\pi(i_{t-1})$.  We choose $i_t$ to be the smallest index
	$i>i_{t-1}$ such that $\pi(i)>v_{t-1}$ and such that an increasing
	subsequence of length $k-t+1$ can still be completed starting from $i$,
	that is, $r_i\ge k-t+1$.  We then set $v_t:=\pi(i_t)$.

	\begin{lemma}\label{lem:leftmost}
		If $\LIS(\pi)=k$, the greedy reconstruction above returns $L^\star(\pi)$.
	\end{lemma}

	\begin{proof}
		We argue by induction on $t$. Suppose $i_1,\dots,i_{t-1}$ have already been
		chosen and coincide with the first $t-1$ positions of the leftmost LIS. Let
		\[
		F_t:=\{i>i_{t-1}:\pi(i)>v_{t-1},\;r_i\ge k-t+1\}.
		\]
		Every index in $F_t$ can be used as the next position of an increasing
		subsequence of total length $k$. For $i\in F_t$, the conditions
		$i>i_{t-1}$ and $\pi(i)>v_{t-1}$ let the prefix $i_1,\dots,i_{t-1}$ be
		continued at $i$, and $r_i\ge k-t+1$ supplies an increasing subsequence
		of length $k-t+1$ starting at $i$, for a total length of
		$(t-1)+(k-t+1)=k$. Conversely, the $t$th position of any
		LIS extending $i_1,\dots,i_{t-1}$ must lie in $F_t$. In particular the
		$t$th position of $L^\star(\pi)$ lies in $F_t$, so the greedy choice
		$i_t=\min F_t$ is at most that position. If it were strictly smaller,
		the full LIS through $i_1,\dots,i_{t-1},i_t$ exhibited above would be
		lexicographically smaller than $L^\star(\pi)$, a contradiction. Thus
		$i_t$ is exactly the $t$th position of $L^\star(\pi)$.
	\end{proof}

	\begin{corollary}\label{cor:test}
		Given a proposal $(\pi,I)$, one can decide whether it is accepted by one
		iteration of Algorithm~\ref{alg:direct-sampler} in $O(n\log\log n)$ time on the word-RAM.
	\end{corollary}

	\begin{proof}
		Compute the values $r_i$ using Lemma~\ref{lem:ri}, in
		$O(n\log\log n)$ time. Their maximum is $\LIS(\pi)$. If this maximum is not
		$k$, reject immediately. Otherwise reconstruct $L^\star(\pi)$ using
		Lemma~\ref{lem:leftmost}. The reconstruction is a single left-to-right scan
		through the permutation and therefore costs $O(n)$ time. Finally compare the
		reconstructed list of positions with $I$, again in $O(n)$ time. The
		total time is $O(n\log\log n)$.
	\end{proof}

	\subsection{Acceptance probability}
	For a partition $\nu$ recall $H_\nu$ and the conjugate $\nu'$ from
	Section~\ref{sec:reparam}. For the empty partition we use the convention
	$\mu_1=0$ and interpret empty products as $1$.

	Any shape with first row of length $k$ can be written as $\lambda=(k,\mu)$,
	where $\mu\vdash m$ is the sub-diagram below the first row. The factor
	$a_k(\mu)$ measures the multiplicative loss in $f^{(k, \mu)}$ caused by the elongation of the first-row hooks relative to a bare row of length $k$. In column $j$, the
	hook length in an isolated row would be $k-j+1$.  Placing $\mu$ underneath
	elongates it by the column height $\mu'_j$ to $k-j+1+\mu'_j$. Columns
	$j>\mu_1$ are unaffected. For $\mu\vdash m$, define
	\[
	a_k(\mu):=
	\begin{cases}
		\displaystyle
		\prod_{j=1}^{\mu_1}\frac{k-j+1}{k-j+1+\mu'_j},
		& \mu_1\le k,\\[1.2em]
		0,& \mu_1>k.
	\end{cases}
	\]
	Define $A_{n,k}:=\E_{\Pl_m}[a_k(\mu)^2]$, where $\Pl_m(\mu)=(f^\mu)^2/m!$ is
	Plancherel measure on partitions of $m$.

	\begin{lemma}\label{lem:factor-largek}
		For $\lambda=(k,\mu)$ with $\mu\vdash m$ and $\mu_1\le k$,
		$f^\lambda=\binom{n}{m}f^\mu a_k(\mu)$.
		Consequently, $|\Omega_{n,k}|=\binom{n}{m}^2m!\,A_{n,k}$.
	\end{lemma}

	\begin{proof}
		For $\lambda=(k,\mu)$, the hooks below the first row are those of $\mu$,
		contributing $H_\mu$. The first-row hooks are $k-j+1+\mu'_j$ for
		$1\le j\le\mu_1$ and $k-j+1$ for $\mu_1<j\le k$, so
		\[
		H_\lambda = H_\mu\,\frac{\prod_{j=1}^{\mu_1}(k-j+1+\mu'_j)}
		{\prod_{j=1}^{\mu_1}(k-j+1)}\cdot k!
		= H_\mu\,\frac{k!}{a_k(\mu)}.
		\]
		The hook-length formula then gives
		\[
		f^\lambda = \frac{n!}{H_\lambda}
		= \frac{n!}{m!\,k!}\cdot\frac{m!}{H_\mu}\cdot a_k(\mu)
		= \binom{n}{m}f^\mu a_k(\mu).
		\]
		For the second claim, Robinson--Schensted identifies $\Omega_{n,k}$ with
		pairs of tableaux of shapes $(k,\mu)$, $\mu\vdash m$,
		$\mu_1\le k$ \cite{Schensted1961}. Since $a_k(\mu)=0$ when $\mu_1>k$,
		the constraint can be dropped in the second equality below, so
		\[
		|\Omega_{n,k}|
		= \sum_{\substack{\mu\vdash m\\\mu_1\le k}}(f^{(k,\mu)})^2
		= \binom{n}{m}^{\!2}\sum_{\mu\vdash m}(f^\mu)^2 a_k(\mu)^2
		= \binom{n}{m}^{\!2}m!\,A_{n,k}.\qedhere
		\]
	\end{proof}

	\begin{corollary}\label{cor:acc-largek}
		The acceptance probability of one iteration of
		Algorithm~\ref{alg:direct-sampler} is $|\Omega_{n,k}|/|\mathcal E_{n,k}|=A_{n,k}$.
	\end{corollary}

	\begin{proof}
		By Proposition~\ref{prop:direct-exact}, every permutation in $\Omega_{n,k}$
		contributes exactly one accepted pair $(\pi,I)$, so the acceptance probability
		equals $|\Omega_{n,k}|/|\mathcal E_{n,k}|$. The result now follows from
		Proposition~\ref{prop:proposal} and Lemma~\ref{lem:factor-largek}.
	\end{proof}

	\begin{lemma}\label{lem:plancherel-first-row-light}
		For every $m\ge 1$, if $\mu$ is drawn from Plancherel measure on partitions
		of $m$, then $\Pp_{\Pl_m}(\mu_1\le 4\sqrt{m})\ge\tfrac{1}{2}$.
	\end{lemma}

	The proof transfers the first row under Plancherel measure to the LIS of a
	uniform permutation and bounds its upper tail by counting increasing
	subsequences and applying Markov's inequality.

	\begin{proof}
		By the Robinson--Schensted correspondence, $\mu_1$ under Plancherel measure has
		the same distribution as $\LIS(\sigma)$ for a uniform random permutation
		$\sigma$ of $[m]$. Fix $r\ge 1$ and let $X_r$ count the increasing
		subsequences of length $r$ in $\sigma$. Each $r$-subset of positions contributes
		with probability $1/r!$, so $\E X_r=\binom{m}{r}/r!$. Since $\LIS(\sigma)\ge r$
		implies $X_r\ge 1$, Markov's inequality gives
		\[
		\Pp_{\Pl_m}(\mu_1\ge r)=\Pp(\LIS(\sigma)\ge r)\le\E X_r=\frac{\binom{m}{r}}{r!}.
		\]
		Using $\binom{m}{r}\le(em/r)^r$ and $r!\ge(r/e)^r$ yields
		$\Pp_{\Pl_m}(\mu_1\ge r)\le(e^2m/r^2)^r$. Taking $r=\lfloor 4\sqrt{m}\rfloor+1>4\sqrt{m}$,
		\[
		\Pp_{\Pl_m}(\mu_1>4\sqrt{m})
		\le\Pp_{\Pl_m}(\mu_1\ge r)
		\le\left(\frac{e^2}{16}\right)^r
		\le\frac{1}{2},
		\]
		which proves the claim.
	\end{proof}

	\begin{lemma}\label{lem:const}
		Fix \(\eta\in(0,1]\).  If \(k\ge \eta n\), then for all sufficiently large
		\(n\),
		\[
		A_{n,k}\ge
		\frac12\exp\!\left(-\frac{4(1-\eta)}{\eta}\right).
		\]
		In particular, if \(k\in\Theta(n)\), then there are constants \(c>0\) and
		\(n_0\) such that \(A_{n,k}\ge c\) for all \(n\ge n_0\).
	\end{lemma}
	\begin{proof}
		Recall that \(m=n-k\) and
		\(A_{n,k}=\E_{\Pl_m}[a_k(\mu)^2]\). The idea is to restrict this
		expectation to the event \(\mu_1\le4\sqrt m\), on which the elongation
		factor \(a_k(\mu)\) is bounded below by a positive constant.
		If \(m=0\), then \(k=n\) and \(A_{n,n}=1\), which is larger than the
		displayed lower bound. Hence assume \(m=n-k\ge 1\). Since \(k\ge\eta n\),
		we have \(\frac{m}{k}=\frac{n-k}{k}\le\frac{1-\eta}{\eta}\).

		Using \(\log(1+x)\le x\), whenever \(\mu_1\le k\) we have
		\[
		-\log a_k(\mu)
		=\sum_{j=1}^{\mu_1}\log\!\left(1+\frac{\mu'_j}{k-j+1}\right)
		\le\sum_{j=1}^{\mu_1}\frac{\mu'_j}{k-j+1}
		\le\frac{m}{k-\mu_1+1},
		\]
		where the last step uses \(\sum_{j=1}^{\mu_1}\mu'_j=m\).

		By Lemma~\ref{lem:plancherel-first-row-light}, the event
		\(\{\mu_1\le 4\sqrt{m}\}\) has probability at least \(1/2\). Since
		\(m\le ((1-\eta)/\eta)k\), for all sufficiently large \(n\) this event
		implies \(\mu_1\le k/2\), and in particular \(\mu_1\le k\). Hence
		\(k-\mu_1+1\ge k/2\). On this event,
		\[
		-\log a_k(\mu)\le\frac{2m}{k}
		\le\frac{2(1-\eta)}{\eta},
		\]
		and therefore $a_k(\mu)^2\ge\exp(-4(1-\eta)/\eta)$.
		It follows that
		\[
		A_{n,k}=\E_{\Pl_m}[a_k(\mu)^2]
		\ge
		\frac{1}{2}\exp\!\left(-\frac{4(1-\eta)}{\eta}\right)
		\]
		for all sufficiently large \(n\).  The final assertion follows by choosing
		any \(\eta>0\) such that \(k\ge\eta n\) for all sufficiently large \(n\).
	\end{proof}

	We now prove Theorem~\ref{thm:intro-largek}, the exact $O(n\log\log n)$
	sampler for the linear regime $k\in\Theta(n)$.

	\begin{proof}[Proof of Theorem~\ref{thm:intro-largek}]
		Exactness is established by Proposition~\ref{prop:direct-exact}. By
		Corollary~\ref{cor:test}, each proposal is tested in $O(n\log\log n)$ time,
		and generating it takes $O(n)$ time. By Corollary~\ref{cor:acc-largek} and
		Lemma~\ref{lem:const}, the acceptance probability is bounded below by a
		positive constant for all sufficiently large $n$, so the expected number of
		trials is $O(1)$. The expected total running time is therefore $O(n\log\log n)$.  Since the
		trials are independent and each succeeds with probability bounded below by
		a constant, the number of trials has geometric tails, so the same bound
		holds with exponentially decaying tail probabilities.
	\end{proof}

	Finally we prove Corollary~\ref{cor:sublinear-k}, the polynomial-time
	bound for $k\in\Omega(n/\log n)$.

	\begin{proof}[Proof of Corollary~\ref{cor:sublinear-k}]
		If $m=0$, then $A_{n,n}=1$ and the claim is trivial.  Hence assume
		$m\ge1$.  Consider the event $\mu_1\le4\sqrt m$, which has Plancherel
		probability at least $1/2$ by
		Lemma~\ref{lem:plancherel-first-row-light}.  Since $k>4\sqrt m$, on this
		event $\mu_1<k$, so the bound
		$-\log a_k(\mu)\le m/(k-\mu_1+1)$ established in the proof of
		Lemma~\ref{lem:const} applies and gives
		$$
		-\log a_k(\mu)
		\le
		\frac{m}{k-\mu_1+1}
		\le
		\frac{m}{k-4\sqrt m+1}.
		$$
		Therefore
		$$
		a_k(\mu)^2
		\ge
		\exp\!\left(-\frac{2m}{k-4\sqrt m+1}\right)
		$$
		on this event. Hence
		$$
		A_{n,k}
		\ge
		\frac12
		\exp\!\left(-\frac{2m}{k-4\sqrt m+1}\right).
		$$
		Since each trial of the rejection sampler costs $O(n\log\log n)$ time, the
		expected running time is
		$$
		O\!\left(
		n\exp\!\left(\frac{2m}{k-4\sqrt m+1}\right)\log\log n
		\right).
		$$
		
		Now fix $\eta>0$ and suppose $n$ is sufficiently large and
		$k\ge \eta n/\log n$. Since $m\le n$, we have
		$4\sqrt m\le4\sqrt n=o(n/\log n)$, so $k>4\sqrt m$ for all sufficiently
		large $n$. Moreover,
		$$
		\frac{m}{k-4\sqrt m+1}
		\le
		\frac{n}{\eta n/\log n-4\sqrt n+1}
		=
		\frac{\log n}{\eta}+O(1).
		$$
		Thus
		$$
		\exp\!\left(\frac{2m}{k-4\sqrt m+1}\right)
		=
		O(n^{2/\eta}),
		$$
		and the expected running time is
		$$
		O\!\left(n^{1+2/\eta}\log\log n\right).
		$$
		For $\eta=1$, this is $O(n^3\log\log n)$.
	\end{proof}

	\section{Sampling permutations with LIS at most \texorpdfstring{$k$}{k}}
	\label{sec:atmost}

	The samplers above extend to the relaxed constraint $\LIS(\pi)\le k$ by the
	same argument, after padding the conjugate with zero parts and replacing the
	ambient coordinate set $\{1,\ldots,n\}$ by $\{0,\ldots,n+k-1\}$.  We first
	record the padded analogue of Lemma~\ref{lem:x-hook}.

	\begin{lemma}[Padded reparameterisation]\label{lem:padded-hook}
	For $\lambda\vdash n$ with $\lambda_1\le k$, pad the conjugate
	$\rho=\lambda'$ with zeros to exactly $k$ weakly decreasing parts
	$\rho_1\ge\cdots\ge\rho_k\ge0$ and set $x_i:=\rho_i+k-i$.  The map
	$\lambda\mapsto x$ is a bijection between
	$\{\lambda\vdash n:\lambda_1\le k\}$ and strictly decreasing sequences
	\[
	n+k-1\ \ge\ x_1>x_2>\cdots>x_k\ \ge\ 0,
	\qquad x_1+\cdots+x_k=N=n+\binom k2,
	\]
	and under this bijection $f^\lambda=n!\,\Delta(x)/\prod_ix_i!$, so the
	target law on padded sequences is again proportional to
	$\wt(x)=\Delta(x)^2\prod_ix_i!^{-2}$.
	\end{lemma}

	\begin{proof}
	The condition $\lambda_1\le k$ means that $\rho$ has at most $k$ parts.
	Padding with zeros and shifting therefore gives the stated bijection,
	exactly as in Section~\ref{sec:general}, but with lower endpoint $0$
	instead of $1$.  It remains to verify the hook-product formula. The
	calculation extends the row products in the proof of
	Lemma~\ref{lem:x-hook} from $\ell$ rows to $k$ rows, and the point is
	that the empty rows can be absorbed into the same formula.

	Let $\ell$ be the number of nonempty
	parts of $\rho$.  For $i\le \ell$, the ordinary row-$i$ hook product in the
	unpadded diagram gives
	\[
	\prod_{c=1}^{\rho_i}h_\rho(i,c)
	=
	\frac{(\rho_i+\ell-i)!}
	{\prod_{i<j\le \ell}(\rho_i-\rho_j+j-i)}.
	\]
	Multiplying numerator and denominator by
	\[
	\prod_{j=\ell+1}^{k}(\rho_i+j-i)
	\]
	turns this into
	\[
	\prod_{c=1}^{\rho_i}h_\rho(i,c)
	=
	\frac{(\rho_i+k-i)!}
	{\prod_{i<j\le k}(\rho_i-\rho_j+j-i)}.
	\]
	For $i>\ell$ the row is empty and contributes $1$, while $\rho_i=0$ and
	\[
	\frac{(\rho_i+k-i)!}
	{\prod_{i<j\le k}(\rho_i-\rho_j+j-i)}
	=
	\frac{(k-i)!}{\prod_{j=i+1}^{k}(j-i)}
	=1.
	\]
	Thus
	\[
	H_\rho
	=
	\frac{\prod_i(\rho_i+k-i)!}
	{\prod_{i<j}(\rho_i-\rho_j+j-i)}.
	\]
	Substituting $x_i=\rho_i+k-i$ gives
	\[
	H_\rho=\frac{\prod_i x_i!}{\Delta(x)}.
	\]
	Since transposition preserves the number of standard Young tableaux,
	\[
	f^\lambda=f^\rho=n!\frac{\Delta(x)}{\prod_i x_i!}.
	\]
	The claimed target weight follows by squaring and dropping the constant
	$(n!)^2$.
	\end{proof}

	We can now prove Corollary~\ref{cor:atmost}.

	\begin{proof}[Proof of Corollary~\ref{cor:atmost}]
	We first treat the general sampler.  By the Robinson--Schensted
	correspondence and Lemma~\ref{lem:padded-hook}, it suffices to sample a strict sequence of
	$k$ coordinates in $\{0,\ldots,n+k-1\}$, with sum $N$, according to the
	same separated weight $\wt$ as before.  Thus only the coordinate range,
	the integer normalisation, and the resulting degree bounds need to be
	adjusted.  The passage from the sampled shape to a permutation is
	Lemma~\ref{lem:shape-to-perm}, whose proof applies verbatim with the
	support $\{\lambda\vdash n:\lambda_1\le k\}$ in place of
	$\{\lambda\vdash n:\lambda_1=k\}$.

	Candidates and suffix values are now drawn from $\{0,1,\ldots\}$.  The
	admissible candidates are $t\in\{0,\ldots,B(p)\}$ with
	$B(\varnothing)=n+k-1$, the suffix ground set is $E_t=\{0,\ldots,t-1\}$,
	and the feasibility window becomes $\binom L2\le s(t)\le Lt-\binom{L+1}2$.
	The lower and upper bounds are the sums of the $L$ smallest and $L$
	largest elements of $\{0,\ldots,t-1\}$, respectively.

	Since coordinate values may now exceed $n$ (they are at most
	$n+k-1\le 2n$), the clearing factor $n!$ in the stage weights is replaced
	by $(n+k-1)!$, that is,
	$G_p(a):=\bigl((n+k-1)!/a!\bigr)^{2}\prod_{i=1}^r(p_i-a)^2$,
	which keeps every completion score a nonnegative integer.

	All coordinate values, and hence all degrees, are bounded by $n+k-1\le 2n$
	instead of $n$, so the determinant polynomials of both implementations
	have degree $O(kn)$, and the primes of
	Section~\ref{sec:direct-time} are chosen larger than
	$d_{\max}:=k(n+k-1)+1\le 2n^2$. The sieve interval
	$(d_{\max},\,n^2\log^2n]$ still contains $\Theta(n^2\log n)$ primes.

	For cumulative scores, the ground set becomes $\{0,\ldots,T\}$, with
	$C_p(-1)=0$, and the binary search runs over $0\le T\le B(p)$.  Since this
	ground set has $T+1$ elements, the empty-subset shortcut in
	Lemma~\ref{lem:fast-hankel-stage} becomes $T<L$.  The terminal $L=0$ stage
	still appends its forced coordinate without invoking the oracle.

	With these changes, Lemma~\ref{lem:score-mass}, Lemma~\ref{lem:cum-mass},
	Theorem~\ref{thm:shape-correct},
	Lemma~\ref{lem:bits-stage} (as $(n+k-1)!\le(2n)!$ still has $O(n\log n)$
	bits and all differences are at most $2n$), and the implementations of
	Sections~\ref{sec:direct-time} and~\ref{sec:complexity} go
	through with the same arguments, with all bounds changing by constant factors only.
	The well-definedness argument of Section~\ref{sec:shape-sampler} also
	applies, with the column shape $\lambda=(1^n)$, that is
	$x=(n+k-1,k-2,k-3,\ldots,0)$, witnessing initial feasibility.  This gives
	expected running time $\tilde O(n^4k^5)$ for the direct implementation and
	$\tilde O(n^2k^4)$ for the fast implementation.

	For the second statement in the corollary with $k \geq 4 \sqrt{n}$, repeat the following trial. Draw $\pi\in S_n$
	uniformly in $O(n)$ time, compute $\LIS(\pi)$ in $O(n\log\log n)$ time
	(Fact~\ref{fact:lis-ending-ranks}), and accept if and only if
	$\LIS(\pi)\le k$.  Conditioned on acceptance, the output is uniform on
	$\Omega_{n,\le k}$.  By the first-moment bound in the proof of
	Lemma~\ref{lem:plancherel-first-row-light}, applied with $m=n$, setting
	$r:=\lfloor4\sqrt n\rfloor+1\le k+1$ gives
	\[
	\Pp\bigl(\LIS(\pi)>k\bigr)\le\Pp\bigl(\LIS(\pi)\ge r\bigr)
	\le\Bigl(\frac{e^2n}{r^2}\Bigr)^{\!r}
	\le\Bigl(\frac{e^2}{16}\Bigr)^{\!r}\le\frac12 ,
	\]
	so the acceptance probability is at least $\tfrac12$, the number of trials
	is geometric, and the expected total time is $O(n\log\log n)$.

	Finally, for $k\le 4\sqrt n$ the general sampler costs
	$\tilde O(n^2k^4)\subseteq\tilde O(n^4)$, while for $k>4\sqrt n$ rejection
	costs $O(n\log\log n)$. The combined bound $\tilde O(n^4)$ follows.
	\end{proof}

	We note the contrast with the constraint $\LIS(\pi)=k$, for which no
	analogous rejection shortcut is available away from the typical values
	close to $2\sqrt n$.
 
	\printbibliography

\end{document}